\documentclass[11pt,reqno]{amsart}

\usepackage{amsmath, amsthm, amscd, amsfonts, amssymb, graphicx, color, mathrsfs, extarrows}
\usepackage{mathscinet}
\usepackage[a4paper,left=3cm,right=3cm,top=2cm,bottom=4cm,bindingoffset=5mm]{geometry}
\usepackage{fancyhdr}
\usepackage{comment}
\usepackage{float}
\usepackage{diagbox}
\usepackage{eurosym}
\usepackage[utf8]{inputenc}
\usepackage{amsmath}
\usepackage{amsfonts}
\usepackage{version}
\usepackage{mathtools}
\usepackage{amscd}
\usepackage{amssymb}
\usepackage{dsfont}
\usepackage{amsthm}
\usepackage{thmtools}
\usepackage{graphicx}
\usepackage{mathrsfs}
\usepackage{todonotes}
\usepackage{xcolor}
\usepackage{todonotes}
\usepackage{setspace}
\usepackage{enumerate}
\usepackage[english]{babel}
\usepackage{xcolor}
\usepackage[colorlinks=true,linkcolor=blue]{hyperref}
\usepackage[all]{hypcap}
\usepackage{newtxmath}
\usepackage{pgfplots}

\newtheorem{theorem}{Theorem}
\newtheorem{proposition}[theorem]{Proposition}
\newtheorem{corollary}[theorem]{Corollary}
\newtheorem{lemma}[theorem]{Lemma}

\theoremstyle{definition}

\newtheorem{remark}[theorem]{Remark}
\newtheorem{example}[theorem]{Example}
\newtheorem{definition}[theorem]{Definition}

\newcommand{\cB}{\mathcal{B}}

\newcommand{\cF}{\mathcal{F}}

\newcommand{\cM}{\mathcal{M}}

\newcommand{\cP}{\mathcal{P}}

\newcommand{\cX}{\mathcal{X}}

\newcommand{\E}{\mathbb{E}}

\newcommand{\N}{\mathbb{N}}
\renewcommand{\P}{\mathbb{P}}

\newcommand{\R}{\mathbb{R}}

\newcommand{\al}{\alpha}
\newcommand{\be}{\beta}
\newcommand{\de}{\delta}
\newcommand{\ep}{\varepsilon}

\newcommand{\ga}{\gamma}
\newcommand{\om}{\omega}
\newcommand{\si}{\sigma}

\newcommand{\rh}{\varrho}

\newcommand{\Om}{\Omega}

\newcommand{\la}{\lambda}

\renewcommand{\d}{{\rm d}}
\newcommand{\eins}{\mathds{1}}
\newcommand{\lo}{{\rm lo}}

\DeclareMathOperator{\VaR}{\rm VaR}

\DeclareMathOperator{\spear}{sp}
\DeclareMathOperator{\cov}{cov}
\DeclareMathOperator{\var}{var}
\DeclareMathOperator{\cor}{cor}
\DeclareMathOperator{\cpl}{cpl}

\DeclareMathOperator{\cnt}{C}
\newcommand{\Cb}{\cnt_{\rm b}}
\DeclareMathOperator{\borel}{B}
\newcommand{\Bb}{\borel_{\rm b}}
\DeclareMathOperator{\pr}{pr}
\newcommand{\es}{\emptyset}
\newcommand{\sm}{\setminus}

\title[Dependence uncertainty and tail risk]{Upper comonotonicity and risk aggregation under dependence uncertainty}

\author{Corrado De Vecchi}
\address{Department of Mathematics, Technical University of Munich, Germany}
\email{corrado.de.vecchi@tum.de}

\author{Max Nendel}
\address{Center for Mathematical Economics, Bielefeld University, Germany}
\email{max.nendel@uni-bielefeld.de}

\author{Jan Streicher}
\address{Center for Mathematical Economics, Bielefeld University, Germany and Landesbank Baden-W\"urttemberg, Stuttgart, Germany}
\email{jan.streicher@uni-bielefeld.de}

\thanks{The second and third author were funded by the Deutsche Forschungsgemeinschaft (DFG, German Research Foundation) -- SFB 1283/2 2021 -- 317210226.\ The third author is grateful for the support of the Landesbank Baden-W\"urttemberg related to this work.}

\date{\today}

\begin{document}
\maketitle

\begin{abstract}
In this paper, we study dependence uncertainty and the resulting effects on tail risk measures, which play a fundamental role in modern risk management.\
We introduce the notion of a regular dependence measure, defined on multi-marginal couplings, as a generalization of well-known correlation statistics such as the Pearson correlation.\ The first main result states that even an arbitrarily small positive dependence between losses can result in perfectly correlated tails beyond a certain threshold and seemingly complete independence before this threshold.\
In a second step, we focus on the aggregation of individual risks with known marginal distributions by means of arbitrary nondecreasing left-continuous aggregation functions. In this context, we show that under an arbitrarily small positive dependence, the tail risk of the aggregate loss might coincide with the one of perfectly correlated losses. A similar result is derived for expectiles under mild conditions.\ In a last step, we discuss our results in the context of credit risk, analyzing the potential effects on the value at risk for weighted sums of Bernoulli distributed losses.


\end{abstract}

\section{Introduction}

The estimation of tail risk is an integral part of modern regulation as it is supposed to quantify the liquidity of financial institutions in the case of rare extreme events.\ 
For banks, credit risk is the most important type of risk.\ 
Loan portfolios or, equivalently, in an insurance context, insurance portfolios typically comprise a large number of customers, each of them representing the risk of a potential loss for the bank or the insurance company.\ While the loss distribution of each customer, individually, may be well-known, e.g., due to ratings or data for classification in risk categories,
the distribution of the aggregate portfolio loss is usually difficult to assess as it requires precise knowledge of the dependence structures within the portfolio, i.e., the \textit{joint} distribution of the loss profiles of all customers.\ However, the modelling and estimation of dependencies within large portfolios represents a complex and daunting task from a mathematical and statistical point of view, respectively, cf.\ \cite{Academicresponse,McNeil2010} and the references therein.\ Therefore, the analysis of the joint distribution is typically confined to a quantification of dependencies using dependence measures, such as the Pearson correlation, Kendall's tau or Spearman's rho.\ As a consequence, the joint distribution is highly susceptible to model uncertainty, and a broad strand of literature has formed around the study of risk bounds under dependence uncertainty, i.e., the study of worst-case scenarios for risk measures of the aggregate loss when the dependence structure among the random variables of interest in not completely specified. 

In this context, the literature distinguishes between two cases, the first one being the case of full dependence uncertainty, i.e., when no information on the dependence structure is available and only the marginal distributions are known.\ Early and seminal contributions in this direction include \cite{Makarov1981} and \cite{Rush1982}, where the authors solve the problem of finding the worst-case value at risk (VaR) for the sum of two random variables, followed by a series of related publications, most of them focusing on the study of worst-case VaR for the sum of an arbitrary number of random losses, cf.\ \cite{MR3145855,DENUIT1999, EMPURU2013,WANG2013}.\ A more recent contribution that studies bounds for spectral risk measures of general aggregation functions using an optimal transport approach is \cite{Ghossoubetal2023}, and general results for tail risk measures of the aggregate sum have been derived in \cite{liu2021theory}.

Although a complete specification of the dependence structure may be difficult to attain in practice, it is realistic to assume that at least some partial knowledge can be inferred from available data.\ This observation has motivated the second strand of literature, which studies risk aggregation problems under partial information on the dependence structure, in addition to the knowledge of the marginal distributions.
One possibility is to work under the assumption that some point wise bounds are available for the copula or the joint distribution of interest, cf.\
\cite{Bignozzietal2015,Embrechts2003}. Another possibility to describe partial knowledge is to assume that the copula is specified on a subset of its domain.\ This case has been studied in \cite{BernVand2015}.\ In the financial industry, partial dependence information is usually described using one or more measures of dependence, the most popular of them being the Pearson correlation, cf.\ \cite{MaiSchererbook2}. In this context, \cite{bernard2023impact} shows that, for the sum of two random variables, a partial specification of the dependence structure via bounds under a dependence measure, such as Pearson's correlation, hardly affects the worst-case VaR. We also refer to \cite{lam2021tail} for an analysis of the worst-case tail behaviour based on a property called tail convexity. 

For an arbitrary number of random variables, risk aggregation problems have been studied under quite specific assumptions on the considered dependence measure.\ For instance, \cite{Bernruvan2017} studied VaR bounds, assuming that the variance of the sum is known to be less or equal than a certain threshold.\ A similar situation in the context of credit risk with exchangeable losses has been considered in \cite{MR4205881}.

In the present paper, we introduce a general notion of dependence measure, comprising essentially all dependence measures used in risk management, cf.\ Example \ref{ex.dep.meas}.\ In our setting, a dependence measure is a map that assigns a value between $-1$ and $1$ to multi-marginal couplings and zero to the product measure.\ Moreover, we say that a dependence measure is regular if it is lower semicontinuous at the product measure with respect to convergence in distribution.\ Given a regular dependence measure, our first main result, Theorem \ref{th: main result 1}, shows that, if an arbitrarily small positive dependence between losses cannot be excluded, their joint probability distribution might have perfectly correlated tails beyond a certain threshold while pretending to have independent marginals before this threshold, so that the true dependence structure in the tails is only revealed after observing an extreme event with perfectly correlated losses.

In Section \ref{se: tail risk measures}, we then apply Theorem \ref{th: main result 1} to tail risk measures of aggregate losses for arbitrary nondecreasing left-continuous aggregation functions.\ Theorem \ref{th: tailequivalent} shows that the worst-case tail risk could potentially be greater or equal than the tail risk under perfect correlation.\ This becomes even more apparent if one considers a coherent tail risk measure, such as the expected shortfall, together with a positively weighted sum as an aggregation function and identically distributed marginals. In this case, the worst-case expected shortfall of the aggregate loss coincides with the weighted sum of the individual expected shortfalls, cf.\ Corollary \ref{cor:tail risk.pos.hom} and the discussion thereafter.\ Using the main result in \cite{MR4216328}, the same holds true for certain expectiles despite the fact that they are not tail risk measures, cf.\ Theorem \ref{thm: expectiles}.\ 

In Example \ref{ex.application}, we illustrate our results in a credit risk context for the value at risk of a sum of Bernoulli distributed random variables for the typical threshold of $99.9\%$, corresponding to a once in a thousand years event.\ Already for $1000$ customers having a correlation less than $0.1$ and $1\%$ probability of default, the worst-case value at risk exceeds the value at risk, given that we have not yet observed a once in a thousand years event, by a factor of $50$.\ This simple example underlines the high sensitivity of tail risk for aggregate positions with respect to imperfect knowledge of the dependence structure.

The rest of the paper is organized as follows.\ In Section \ref{sc: Dependence meausres}, we introduce the notion of a regular dependence measure and state the fist main result, Theorem \ref{th: main result 1}, where we construct the previously described joint distribution.\ In Section \ref{ sc: Monotonicity}, we investigate monotonicity properties of the joint distribution in lower orthant order and PQD order with respect to the dependence constraint.\ In Section \ref{se: tail risk measures}, we transfer the results from Section \ref{sc: Dependence meausres} to the value at risk, cf.\ Theorem \ref{th: var.equality} and, in a second step, to arbitrary tail risk measures and expectiles, cf.\ Theorem \ref{th: tailequivalent} and Theorem \ref{thm: expectiles}. The Appendix \ref{app.A} contains two standard results on multi-marginal couplings and the convergence of copulas, which we state and prove for the reader's convenience.\\

\textbf{Notation:} Let $(S,\tau)$ be a Polish space.\ Then, $\cB(S):=\si(\tau)$ denotes the Borel $\si$-algebra on $S$ and $\cP(S)$ denotes the set of all probability measures on $\cB(S)$. The space of all bounded continuous functions $S\to \R$ is denoted by $\Cb(S)$ and the space of all bounded Borel measurable functions $S\to \R$ by $\Bb(S)$.\ Throughout, the set $\cP(S)$ is endowed with the weak topology.\ Recall that the weak topology is metrizable and that a sequence $(\mu^k)_{k\in \N}\subset \cP(S)$ converges to $\mu\in \cP(S)$ in the weak topology if and only if
\[
\lim_{k\to \infty}\int_S f(x)\,\mu^k(\d x)= \int_S f(x)\,\mu(\d x)\quad\text{for all }f\in \Cb(S).
\]
In this case, we write $\mu^k\to \mu$ as $k\to \infty$. 

For $\nu\in \cP(\R)$, we use the notation $F_\nu(a):=\nu\big((-\infty,a]\big)$ for $a\in \R$ and
\[
F_\nu^{-1}(u):=\inf\big\{ a\in \R \,\big|\, F_\nu(a)\geq u\big\}\quad\text{for }u\in (0,1).
\]
For two random variables $X$ and $Y$ on some probability space $(\Om,\cF,\P)$, we write $X\stackrel{\rm d}{=}Y$ if $X$ and $Y$ have the same distribution under $\P$.\ Moreover, for a random variable $X$ on  some probability space $(\Om,\cF,\P)$ with distribution $\nu\in \cP(\R)$,
$$\VaR_\P^\alpha (X):= \inf \big\{ a\in \R\,\big|\,\P(X>a)\leq 1-\alpha\big\}=F_\nu^{-1}(\al)$$ 
denotes the left-continuous version of the \textit{value at risk} (VaR) at level $\alpha\in(0,1)$.\
Last but not least, we use the convention $\frac{0}{0}:=0$.

\section{Multivariate dependence measures and first main result} \label{sc: Dependence meausres}

Let $n\in \N$.\ In this section, we introduce the notion of a dependence measure for general Borel probability measures on $\R^n$.\ For $i=1,\ldots, n$, let $\pr_i\colon \R^n\to \R$ be the $i$-th coordinate projection.

For $\mu_1,\ldots,\mu_n\in \cP(\R)$, we denote by $\cpl(\mu_1,\ldots,\mu_n)$ the set of all $\pi\in \cP(\R^n)$ with
$$
\pi\circ\pr_i^{-1}=\mu_i\quad \text{for all }i=1,\ldots, n.
$$
The elements of $\cpl(\mu_1,\ldots,\mu_n)$ are called \textit{multi-marginal couplings}.\ We refer to \cite{MR0761565} and \cite{MR3423275} for an overview on multi-marginal optimal transport problems and applications. Moreover, we refer to \cite{ambrosio2008gradient} and \cite{villani2008optimal} for a survey on classical optimal transport and related topics.\ Observe that, for all $\mu_1,\ldots,\mu_n\in \cP(\R)$, the product measure $\mu_1\otimes \cdots \otimes \mu_n$ is always an element of $\cpl(\mu_1,\ldots, \mu_n)$.\ In the sequel, we also use the notation
$$
\cpl(\mu)=\cpl(\mu_1,\ldots,\mu_n) \quad\text{for }\mu=(\mu_1,\ldots, \mu_n)\in \cP(\R)^n
$$
and, for a nonempty set $\cM\subset \cP(\R)$, we define
$$
\cpl^n(\cM):=\bigcup_{\mu\in \cM^n} \cpl(\mu).
$$

\begin{definition}\label{def.dependence}
Let $\cM\subset \cP(\R)$ be nonempty.
\begin{enumerate}
\item[a)] A map $\rh\colon \cpl^n(\cM)\to [-1,1]$ is called a \textit{dependence measure} if 
$$
\rh\big(\mu_1\otimes\cdots\otimes \mu_n\big)=0\quad \text{for all }\mu_1,\ldots,\mu_n\in \cM.
$$
\item[b)] We say that a dependence measure $\rh\colon \cpl^n(\cM)\to [-1,1]$ is \textit{regular} if, for all $\mu_1,\ldots,\mu_n\in \cM$ and any sequence $(\pi^k)_{k\in \N}\subset \cpl(\mu_1,\ldots, \mu_n)$ with $\pi^k\to \mu_1\otimes\cdots\otimes \mu_n$ as $k\to \infty$, it follows that
\[
\limsup_{k\to \infty} \rh(\pi^k)\leq 0.
\]
\end{enumerate}
\end{definition}

Note that we do not exclude the case $n=1$. However, in this case, the only dependence measure is $\rh\equiv 0$. If $n=2$, we use the terminology \textit{bivariate dependence measure}.\ 
In situations, where we explicitly consider the case $n>2$, we also use the expression \textit{multivariate dependence measure}.

\begin{remark}
Let $\cM\subset \cP(\R)$ be nonempty.\ Then, we say that a dependence measure $\rh\colon \cpl^n(\cM)\to [-1,1]$ is \textit{symmetric} if, for all multi-marginal couplings $\pi\in \cpl^n(\cM)$ and any permutation $\si\colon \{1,\ldots, n\}\to \{1,\ldots, n\}$, 
\[
\rh(\pi)=\rh\big(\pi\circ (x\mapsto x_\si)^{-1}\big),
\]
where $x_\si:=(x_{\si(1)},\ldots, x_{\si(n)})\in \R^n$ for all $x\in \R^n$.
\end{remark}

\begin{example}\label{ex.abs.dep.meas}
Let $\cM\subset \cP(\R)$ be nonempty.
\begin{enumerate}
 \item[a)] Let $\rh\colon \cpl(\cM)\to [-1,1]$ be a dependence measure with  
\[
\lim_{k\to \infty} \rh(\pi^k)= 0
\]
for any sequence $(\pi^k)_{k\in \N}\subset \cpl(\mu_1,\ldots, \mu_n)$ with $\pi^k\to \mu_1\otimes\cdots\otimes \mu_n$ as $k\to \infty$ and $\mu_1,\ldots,\mu_n\in \cM$. Then, both, $\rh$ and $-\rh$ are regular dependence measures.
\item[b)] Let $\rh_j\colon \cpl^n(\cM)\to [-1,1]$ be a regular dependence measure for $j=1,\ldots, \ell$ with $\ell\in \N$ and $D\colon [-1,1]^\ell\to [-1,1]$ be a nondecreasing function with $$D(0)=\lim_{\de\downarrow 0} D(\de \textbf{1})=0,$$ where $\textbf{1}$ denotes the $\ell$-dimensional vector consisting of only ones.\ Then, the map $\rh\colon \cpl^n(\cM)\to [-1,1]$, given by
 \[
 \rh(\pi):=D\big(\rh_1(\pi),\ldots, \rh_\ell(\pi)\big)\quad \text{for }\pi\in \cpl^n(\cM),
 \]
 defines a regular dependence measure. In fact, let $\mu_1,\ldots, \mu_n\in \cM$.\ Then, 
 \[
  \rh(\mu_1\otimes\cdots\otimes \mu_n)=D\big(\rh_1(\mu_1\otimes \cdots\otimes \mu_n),\ldots, \rh_\ell(\mu_1\otimes\cdots\otimes \mu_n)\big)=D(0)=0.
 \]
 Now, let $(\pi^k)_{k\in \N}\subset \cpl^n(\mu)$ with $\pi^k\to \mu_1\otimes \cdots\otimes \mu_n$ as $k\to \infty$ and $\ep>0$.\ Then, there exists some $\de>0$ such that $D(\de \textbf{1})<\ep$. Since $\rh_1,\ldots, \rh_\ell$ are regular, there exists some $k_0\in \N$ such that
 \[
 \sup_{k\geq k_0} \rh_j(\pi^k)\leq \de\quad\text{for }j=1,\ldots, \ell.
 \]
 Since $D$ is nondecreasing, it follows that 
 \begin{align*}
  \sup_{k\geq k_0} \rh(\pi^k)&=\sup_{k\geq k_0} D\big(\rh_1(\pi^k), \ldots, \rh_\ell(\pi^k)\big)\leq D\bigg(\sup_{k\geq k_0}\rh_1(\pi^k),\ldots, \sup_{k\geq k_0}\rh_\ell(\pi^k)\bigg)\\
  &\leq D(\de \textbf{1})<\ep.
 \end{align*}
 Explicit examples for the function $D$ are given by the following constructions.
 \begin{itemize}
  \item \textit{Weighted sum.}\ For fixed weights $w_1,\ldots, w_\ell\in [0,1]$ with $\sum_{j=1}^\ell w_j\leq 1$, let
  $$D(r):=\sum_{j=1}^\ell w_j r_j\quad\text{for }r=(r_1,\ldots, r_\ell)\in [-1,1]^\ell.$$
  \item \textit{Weighted minimum/maximum.}\ For fixed weights $w_1,\ldots, w_\ell\in [0,1]$, let
  $$D(r):=\min_{j=1,\ldots,\ell} w_j r_j\quad\text{or}\quad D(r):=\max_{j=1,\ldots,\ell} w_j r_j\quad \text{for }r\in [-1,1]^\ell.$$
 \end{itemize}
 \item[c)] Let $n\geq0$ and $\rh_{ij}\colon \cpl^2(\cM)\to [-1,1]$ be a regular bivariate dependence measure for all $i,j\in \{1,\ldots, n\}$ with $i\neq j$.\ For $\pi\in \cP(\R^n)$, let $$\pi_{ij}:=\pi\circ(\pr_i,\pr_j)^{-1}\in \cP(\R^2)\quad \text{for }i,j=1,\ldots, n\text{ with }i\neq j.$$ By definition, $\pi_{ij}\in \cpl^2(\cM)$ for $\pi\in \cpl^n(\cM)$ and $i,j=1,\ldots, n$ with $i\neq j$, and 
 $$
  \cpl^n(\cM)\to [-1,1],\; \pi\mapsto \rh_{ij}(\pi_{ij})
 $$
 is a regular dependence measure for all $i,j=1,\ldots,n$ with $i\neq j$.\ In fact, for all $\mu_1,\ldots, \mu_n\in \cM$ and $i,j=1,\ldots, n$ with $i<j$,
 \[
 (\mu_1\otimes\cdots\otimes \mu_n)_{ij}=\mu_i\otimes \mu_j
 \]
 and $\pi_{ij}^k\to \mu_i\otimes \mu_j$ as $k\to \infty$ for any sequence $(\pi^k)_{k\in \N}\subset \cpl^n(\cM)$ with $\pi^k\to \mu_1\otimes \cdots \otimes\mu_n$ as $k\to \infty$.
  
   Hence, by part a), for any nondecreasing function $D\colon [-1,1]^{n(n-1)}\to [-1,1]$ with $D(0)=\lim_{\de\downarrow 0} D(\de \textbf{1})=0$, where $\textbf{1}$ denotes the $n(n-1)$-dimensional vector consisting of only ones, $\rh\colon \cpl^n(\cM)\to [-1,1]$, given by
 \begin{equation}\label{eq.aggregation}
 \rh(\pi):=D\Big(\big(\rh_{ij}(\pi_{ij})\big)_{i\neq j}\Big)\quad \text{for }\pi\in \cpl^n(\cM),
 \end{equation}
 defines a regular dependence measure.
\end{enumerate}
\end{example}

In the following example, we present prominent bivariate regular dependence measures, which can be used as building blocks for multivariate dependence measures, as described in the previous example.

\begin{example}\label{ex.dep.meas}
\begin{enumerate}
 \item[a)] Let $\cM$ be the set of all $\mu\in \cP(\R)$ with $\int_\R x^2\,\mu(\d x )<\infty$, and recall that we use the convention $\frac00=0$.\ For $\mu\in \cM$, let
 \[
 \var(\mu):=\int_{\R^2} x^2\,\mu(\d x )-\bigg(\int_\R x\,\mu(\d x )\bigg)^2.
 \]
  Then, for $\mu,\nu\in \cM$ and $\pi\in \cpl(\mu,\nu)$, the \textit{Pearson correlation}
 \[
 \cor(\pi):=\frac{\cov (\pi)}{\sqrt{\var(\mu)\var(\nu)}}
 \]
 defines a symmetric and regular dependence measure, where
 \[
 \cov(\pi):=\int_{\R^2} xy \,\pi(\d x,\d y)-\int_\R x\,\mu(\d x)\int_\R y\,\nu(\d y).
 \]
 Clearly, $\cor$ is a dependence measure.\ In order to prove the regularity of $\cor$, let $\mu,\nu\in \cP(\R)$ and $(\pi^k)_{k\in \N}\subset \cpl(\mu,\nu)$ with $\pi^k\to \mu\otimes \nu$ as $k\to \infty$.\ Then,
 \[
   \int_{\R^2} x^2+y^2 \,\pi^k(\d x,\d y)=\int_\R x^2\,\mu(\d x)+\int_\R y^2\,\nu(\d y)\quad\text{for all }k\in \N.
 \]
 Hence, by \cite[Definition 6.8]{villani2008optimal},
 \[
 \int_{\R^2} xy\,\pi^k(\d x,\d y)\to \int_\R x \,\mu(\d x) \int_\R y \,\nu(\d y)\quad \text{as }k\to \infty,
 \]
 so that $\lim_{k\to \infty}\cor(\pi^k)= 0$. 
 \item[b)] Let $\cM$ be the set of all Borel probability measures on $\R$ with continuous distribution function.\ By Sklar's theorem, for $\mu,\nu\in \cM$ and $\pi\in \cpl(\mu,\nu)$, there exists a unique copula $c\colon [0,1]^2\to [0,1]$ with
 \[
 \pi\big((-\infty, a]\times (-\infty,b]\big)=c\big(F_\mu(a),F_\nu(b)\big)\quad\text{for all }a,b\in \R.
 \]
  Then, for $\pi\in \cpl(\cM)$ with copula $c$, \textit{Spearman's rho}
 \[
  \spear(\pi):=12\int_{[0,1]^2} c(u,v)\,\d(u,v)-3
 \]
 defines a symmetric and regular dependence measure.\ Since $\int_{[0,1]^2} uv\,\d(u,v)=\frac14$, it follows that $\spear(\mu\otimes \nu)=0$ for all $\mu,\nu\in \cP(\R)$.\ Now, let $\mu,\nu\in \cP(\R)$ and $(\pi^k)_{k\in \N}\subset \cpl(\mu,\nu)$ with $\pi^k\to \mu\otimes \nu$ as $k\to \infty$. Since $F_\mu$ and $F_\nu$ are continuous, by Lemma \ref{lem.unif.convergence.couplings}, it follows that
 \[
 \sup_{u,v\in [0,1]}\big| c^k(u,v)-uv\big|\to 0\quad\text{as }k\to \infty,
 \]
  where $c^k$ is the copula related to $\pi^k$. Hence,
 \[
  \lim_{k\to \infty}\int_{[0,1]^2} c^k(u,v)\,\d(u,v)=\int_{[0,1]^2} uv\,\d(u,v),
 \]
 so that $\lim_{k\to \infty}\spear(\pi^k)= 0$.
 \item[c)] Consider the same situation as in the previous example.\ Then, for $\pi\in \cpl(\cM)$ with copula $c$, \textit{Kendall's tau}
 \[
 \tau(\pi):=4\int_{[0,1]^2} c(u,v)\,\d c(u,v)-1
 \]
 defines a symmetric and regular dependence measure.\ In fact, since $$\int_{[0,1]^2} uv\,\d(u,v)=\frac14,$$ it follows that $\tau(\mu\otimes \nu)=0$ for all $\mu,\nu\in \cP(\R)$. In order to prove the regularity, let $\mu,\nu\in \cP(\R)$ and $(\pi^k)_{k\in \N}\subset \cpl(\mu,\nu)$ with $\pi^k\to \mu\otimes \nu$ as $k\to \infty$.\ Again, since $F_\mu$ and $F_\nu$ are continuous, by Lemma \ref{lem.unif.convergence.couplings}, it follows that
 \[
 \sup_{u,v\in [0,1]}\big| c^k(u,v)-uv\big|\to 0\quad\text{as }k\to \infty,
 \]
 where $c^k$ is the copula related to $\pi^k$.\ Hence, by H\"older's inequality,
 \[
 \bigg|\int_{[0,1]^2} c^k(u,v)\,\d c^k(u,v) -\frac14\bigg|\leq \sup_{u,v\in [0,1]} \big|c^k(u,v)-uv\big|+\bigg|\int_{[0,1]^2} uv\,\d c^k(u,v)- \frac{1}4\bigg|\to 0
 \]
 as $k\to \infty$.\ We have therefore shown that $\lim_{k\to \infty}\tau(\pi^k)= 0$.
\end{enumerate}

\end{example}

Let $\cM$ be nonempty, $(\Om,\cF,\P)$ be an atomless probability space, and $\rh\colon \cpl^n(\cM)\to [0,1]$ be a dependence measure.\ Then, for random variables $X_1,\ldots,X_n$ on $\Om$, we use the notation
\[
\rh(X_1,\ldots, X_n):=\rh \big(\P\circ (X_1,\ldots,X_n)^{-1}\big).
\]
Recall that a probability space $(\Om,\cF,\P)$ is atomless if and only if there exists a uniformly distributed random variable $U\colon \Om \to (0,1)$.\ The following theorem is the first main result.

\begin{theorem} \label{th: main result 1}
Let $(\Om,\cF,\P)$ be a probability space, $U\colon \Om\to (0,1)$ uniformly distributed, $n\in \N$, $\cM\subset \cP(\R)$ nonempty, and $\rh\colon \cpl^n(\cM)\to [-1,1]$ a regular dependence measure.\ Then, for all $\de\in (0,1)$ and $\mu_1,\ldots, \mu_n\in \cM$, there exist $\ga\in (0,1)$ and random variables $X_1,\ldots X_n$ on $\Om$ such that
\begin{enumerate}
\item[(i)] $X_i\sim \mu_i$ for all $i=1,\ldots, n$,
\item[(ii)] $\rh(X_1,\ldots, X_n)\leq \de$,
\item[(iii)] $X_1,\ldots, X_n$ are conditionally independent on the event $\{U\leq \ga\}$,
\item[(iv)] for all $i=1,\ldots, n$,
\[
\P\big(\big\{X_i> F_{\mu_i}^{-1}(\ga)\big\}\cap \{U\leq \ga\}\big)=0\quad\text{and}\quad X_i\eins_{\{U> \ga\}}=F_{\mu_i}^{-1}(U)\eins_{\{U> \ga\}}.
\]
\end{enumerate}
\end{theorem}

\begin{proof}
Let $\mu_1,\ldots, \mu_n\in \cM$ and $Y_i:=F_{\mu_i}^{-1}(U)$ for $i=1,\ldots, n$.\ For $\ga\in (0,1)$, we define $C^{\ga}:=\{U\leq \ga\}$. Since $U$ is uniformly distributed, $(C^{\ga},\cF\cap C^{\ga},\P^{\ga})$ is atomless, where
\[
\P^{\ga}(A):=\frac{\P(A)}\ga\quad\text{for all }
A\in \cF\cap C^{\ga}\]
and $\cF\cap C^{\ga}:=\{A\cap C^{\ga}\,|\, A\in \cF\}$.\ Hence, there exist $Z_1^\ga,\ldots, Z_n^\ga\colon C^{\ga}\to \R$ that are independent under $\P^{\ga}$ and satisfy
\[
\P^{\ga}(Z_i^\ga\leq a)=\frac{\P(\{Y_i\leq a\}\cap C^{\ga})}\ga\quad\text{for all }a\in \R\text{ and }i=1,\ldots, n.
\]
For $i=1,\ldots, n$, we thus define $X_i^\ga\colon \Om\to \R$ by
\[
X_i^\ga(\om):=\begin{cases} Z_i^\ga(\om),&\om\in C^{\ga},\\
Y_i(\om),&\om\in\Om\sm C^{\ga}.
\end{cases}
\]
Then,
\begin{align*}
\P(X_i^\ga\leq a)&=\ga \P^{\ga}(Z_i^\ga\leq a)+\P\big(\{Y_i\leq a\}\cap (\Om\sm C^{\ga})\big)\\
&= \P(\{Y_i\leq a\}\cap C^{\ga})+\P\big(\{Y_i\leq a\}\cap (\Om\sm C^{\ga})\big)\\
&=\P(Y_i\leq a)=F_{\mu_i}(a)
\end{align*}
for all $a\in \R$ and $i=1,\ldots, n$. Therefore, $X_i^\ga\sim \mu_i$ and
\begin{align*}
\P\big(\big\{X_i^\ga> F_{\mu_i}^{-1}(\ga)\big\}\cap \{U\leq \ga\}\big)&=\ga\P^{\ga}\big(Z_i^\ga>F^{-1}_{\mu_i}(\ga)\big)=\P\big(\big\{Y_i>F^{-1}_{\mu_i}(\ga)\big\}\cap C^{\ga}\big)\\
&=\P\big(\big\{F_{\mu_i}^{-1}(U)>F^{-1}_{\mu_i}(\ga)\big\}\cap\{U\leq \ga\}\big)\\
&=\P(\es)=0\quad\text{for }i=1,\ldots, n.
\end{align*}
We define $\pi^{\ga}:=\P\circ (X^\ga)^{-1}\in \cpl(\mu_1,\ldots, \mu_n)$ for all $\ga\in (0,1)$, and it remains to show that $\pi^\ga\to \mu_1\otimes \cdots \otimes \mu_n$ as $\ga\to 1$. To that end, let $(\ga^k)_{k\in \N}\subset (0,1)$ with $\ga^k\to 1$. Since $\cpl(\mu_1,\ldots, \mu_n)$ is weakly compact, cf.\ Lemma \ref{lem.compactness.couplings}, there exists a subsequence $(\ga^{k_l})_{l\in \N}$ with $\pi^{\ga^{k_l}}\to \pi\in \cpl(\mu_1,\ldots, \mu_n)$ as $l\to \infty$.\ On the other hand, by dominated convergence, for all $f_1,\ldots, f_n\in \Bb(\R)$ with $f_i\geq 0$ for $i=1,\ldots, n$,
\begin{align}
\notag \E\big( f_1(X_1^\ga)\cdots f_n(X_n^\ga)\big)&=\ga \E^\ga\big( f_1(Z_1^\ga)\cdots f_n(Z_n^\ga)\big)+\E\big( f_1(Y_1)\cdots f_n(Y_n) \eins_{\{U>\ga\}}\big)\\
\notag & = \ga \E^\ga\big( f_1(Z_1^\ga)\big)\cdots \E^\ga\big(f_n(Z_n^\ga)\big)+\E\big( f_1(Y_1)\cdots f_n(Y_n) \eins_{\{U>\ga\}}\big)\\
\notag & =\ga^{1-n} \prod_{i=1}^n \E\big( f_i(Y_i)\eins_{\{U\leq \ga\}}\big)+\E\big( f_1(Y_1)\cdots f_n(Y_n) \eins_{\{U>\ga\}}\big)\\
&\to \prod_{i=1}^n \E\big( f_i(Y_i)\big)\quad\text{as }\ga\to 1, \label{eq.dist.coupling.ga}
\end{align}
where $\E^\ga$ denotes the expected value with respect to $\P^\ga$. Since $\pi^{\ga^{k_l}}\to \pi$ as $l\to \infty$, this implies that
\[
\int_{\R^n} f_1(x_1)\cdots f_n(x_n)\,\pi(\d x_1,\ldots, \d x_n) = \prod_{i=1}^n\int_\R f_i(x_i) \,\mu_i(\d x_i)
\]
for all $f_1,\ldots, f_n\in \Cb(\R)$ with $f_i\geq0$ for $i=1,\ldots, n$.\ Using monotone convergence together with the fact that $\eins_{U}$ is bounded and lower semicontinuous for all open sets $U\subset \R$, it follows that
\[
\pi(U_1\times \cdots \times U_n)=\mu_1(U_1)\cdots \mu_n(U_n)
\]
for all open sets $U_1,\ldots, U_n\subset \R$.\ Since the system of all open subsets of $\R$ is an intersection-stable generator of the Borel $\si$-algebra, it follows that $\pi=\mu_1\otimes \cdots \otimes\mu_n$ by Dynkin's lemma.\ Hence, $\pi^{\ga^{k_l}}\to \mu_1\otimes \cdots \otimes \mu_n$ as $l\to \infty$, which implies that $\pi^{\ga}\to \mu_1\otimes \cdots \otimes \mu_n$ as $\ga\to 1$, since we have shown that every subsequence has a further subsequence that converges to the same limit $\mu_1\otimes \cdots \otimes \mu_n$. Since $\rh$ is regular, it follows that
\[
\limsup_{\ga\to 1}\rh(\pi^{\ga})\leq 0.
\]
The proof is complete.
\end{proof}

\begin{remark}\label{rem.mult.dependence}
 Although Theorem \ref{th: main result 1} is formulated in terms of a single constraint 
 \begin{equation}\label{eq.delta.constraint}
 \rh(X_1,\ldots, X_n)\leq \de,
 \end{equation}
 it can easily be extended to include, for example, multiple constraints of the form 
 \begin{equation}\label{eq.delta.mult.constraint}
 \rh_j(X_1,\ldots, X_n)\in \big[\de_j^-,\de_j^+\big]\quad\text{for }j=1,\ldots,\ell
 \end{equation}
 with $\ell\in \N$, dependence measures $\rh_j\colon \cpl^n(\cM)\to [-1,1]$ with $\lim_{k\to \infty} \rh_j(\pi^k)=0$ for all sequences $\pi^k\in \cpl(\mu_1,\ldots, \mu_n)$ and $\mu_1,\ldots, \mu_n\in \cM$ with $\pi^k\to \mu_1\otimes\cdots\otimes \mu_n$ as $k\to \infty$, and $\pm\de_j^\pm\in (0,1)$ for $j=1,\ldots, \ell$. In fact, let $\de:=\min_{j=1,\ldots, \ell} \big(\de_j^+\wedge (-\de_j^-)\big)$ and define
 \begin{equation}\label{eq.rh.mult.constraints}
 \rh(\pi):= \de \max_{j=1,\ldots \ell} \bigg(\frac{\rh_j(\pi)}{\de_j^+}\vee \frac{ \rh_j(\pi)}{\de_j^-}\bigg)\quad\text{for all }\pi\in \cpl^n(\cM).
 \end{equation}
 Then, by Example \ref{ex.abs.dep.meas} a) \& b), $\rh\colon \cpl^n(\cM)\to [-1,1]$ is a regular dependence measure. Since $\de_j^+>0$ and $\de_j^-<0$, the $\ell$ constraints formulated in \eqref{eq.delta.mult.constraint} are equivalent to the single constraint \eqref{eq.delta.constraint} for $\rh$ given by \eqref{eq.rh.mult.constraints}. 

 Similarly, for regular dependence measures $\rh_j\colon \cpl^n(\cM)\to [-1,1]$ and $\de_j\in (0,1)$ for $j=1,\ldots, \ell$, one can consider multiple one-sided constraints of the form
\begin{equation}\label{eq.constraints.onesided}
 \rh_j(X_1,\ldots, X_n)\leq \de_j\quad \text{for }j=1,\ldots, \ell,
 \end{equation}
 by setting $\de:=\min_{j=1,\ldots, \ell} \de_j$ and $\rh(\pi):=\max_{j=1,\ldots, \ell}  \frac{\de}{\de_j}\rh_j(\pi)$ for $\pi\in \cpl^n(\cM)$.
\end{remark}

\begin{remark}\label{rem.rap}
    We briefly discuss Theorem \ref{th: main result 1} from a slightly different angle and in view of the existing literature.\ It is worth noting that, in the formulation of Theorem \ref{th: main result 1}, the underlying common risk factor $U$ is fixed a priori, i.e., for \textit{every} given uniformly distributed risk factor $U$, every $\de\in (0,1)$, and all $\mu_1,\ldots, \mu_n\in \cM$, we find $\ga\in (0,1)$ and a random vector $X=X^\ga$ satisfying the properties (i) - (iv).\ Clearly, the uniformly distributed risk factor $U$ can be replaced by any given random variable $W$ with continuous distribution function $F$, replacing $\ga$ by $F^{-1}_{\rm R}(\ga):=\inf\{a\in \R \,|\, F(a)>\ga\}$.\footnote{Indeed, since $F$ is continuous, $U:=F(W)$ is uniform, cf.\ \cite[Lemma A.25]{MR3859905}, and $w\leq F_{\rm R}^{-1}(\ga)$ is equivalent to $F(w)=\P(W<w)\leq \ga$ for all $w\in \R$ and $\ga\in (0,1)$.\ Hence, $W\leq F_{\rm R}^{-1}(\ga)$ if and only if $U=F(W)\leq \ga$.}  
 
    If, instead, we allow ourselves to choose the uniformly distributed random factor $U$ freely, the family $(X^\ga)_{\ga\in (0,1)}$ of random vectors can be constructed more explicitly.\ To that end, let $(V_1,\ldots ,V_n,U)$ be a vector of independent random variables with uniform distribution on $(0,1).$ For given marginal distributions $\mu_1,...,\mu_n \in \cP(\R) $ and $\gamma \in (0,1)$, we define the $i$-th component of the random vector $X^\gamma$ as
   \begin{equation}\label{eq: rap}
       X^\gamma_i=F^{-1}_{\mu_i}\big(\gamma V_i \mathds{1}_{\{U\leq\gamma\}} + U \mathds{1}_{\{U>\gamma\}}\big)\quad \text{for }i=1,\ldots, n.
   \end{equation}
   Then, $X^\ga\sim \pi^{\ga}$ with $\pi^{\ga}\in \cpl(\mu_1,\ldots, \mu_n)$ as in the proof of Theorem \ref{th: main result 1} and, by definition of $X_i^\ga$, $X_i^\ga\eins_{\{U>\ga\}}=F_{\mu_i}^{-1}(U)\eins_{\{U>\ga\}}$ for $i=1,\ldots, n$.\ In fact, using the notation from the proof of Theorem \ref{th: main result 1}, the independence of $(V_1,\ldots, V_n,U)$, and the equality derived along the first three lines of \eqref{eq.dist.coupling.ga}, for all $f_1,\ldots, f_n\in \Bb(\R)$ with $f_i\geq 0$ for $i=1,\ldots, n$,
\begin{align}
\notag  \E\big( f_1(X_1^\ga)\cdots f_n(X_n^\ga)\big)&=\ga \prod_{i=1}^n \E\Big( f_i\big(F_{\mu_i}^{-1}(\ga V_i)\big)\Big)+\E\Big( f_1\big(F_{\mu_1}^{-1}(U)\big)\cdots f_n\big(F_{\mu_n}^{-1}(U)\big) \eins_{\{U>\ga\}}\Big)\\
\notag & = \ga \prod_{i=1}^n\int_0^1 f_i\big(F_{\mu_i}^{-1}(\ga v_i)\big)\,\d v_i +\E\big( f_1(Y_1)\cdots f_n(Y_n) \eins_{\{U>\ga\}}\big)\\
\notag & =\ga^{1-n} \prod_{i=1}^n \int_0^\ga f_i\big(F_{\mu_i}^{-1}(u)\big)\,\d u+\E\big( f_1(Y_1)\cdots f_n(Y_n) \eins_{\{U>\ga\}}\big)\\
\notag & =\ga^{1-n} \prod_{i=1}^n \E\big( f_i(Y_i)\eins_{\{U\leq \ga\}}\big)+\E\big( f_1(Y_1)\cdots f_n(Y_n) \eins_{\{U>\ga\}}\big)\\
\label{eq.coupling.xgamma} &=\int_{\R^d}f_1(x_1)\cdots f_n(x_n)\,\pi^{\ga}(\d x_1,\ldots, \d x_n).
\end{align}
 Since the system 
 $\big\{(-\infty,a]\colon a\in \R^n\big\}$ 
 is an intersection-stable generator of the Borel $\si$-algebra on $\R^n$, it follows that $X^\ga\sim \pi^{\ga}$. Hence, when discussing distributional properties of the family $(X^\ga)_{\ga\in (0,1)}$, we will often use the more convenient representation \eqref{eq: rap} of the $i$-th coordinate of $X^\ga$ for $\ga\in (0,1)$ and $i=1,\ldots, n$.
 
   Observe that $X^\gamma$ is an upper comonotonic random vector in the sense of \cite{Cheug2009uc}.\ More specifically, $X^\gamma$ has a comonotonic support in the upper set $\prod^{n}_{i=1}\big(F^{-1}_{\mu_i}(\gamma), \infty\big)$.
   Furthermore, $X^\gamma$ is an $\alpha$-concentrated random vector in the sense of \cite{WangZitikis2020}, for any $\alpha\in[\gamma,1)$, and the common $\alpha$-tail events can be described as $\{U>\alpha\}$, on which $U$ is the uniform random variable considered in the representation of $X^\gamma$ given in \eqref{eq: rap}.\ 
   Theorem~\ref{th: main result 1} together with Remark \ref{rem.mult.dependence} then states that, even under multiple constraints on possibly different dependence measures, one can always construct an upper comonotonic vector as in \eqref{eq: rap} that satisfies these constraints, as this is true regardless of the considered marginal distributions.\
Random vectors of the form \eqref{eq: rap}, for sufficiently high values of $\gamma$, aim to describe those situations in which extremely large losses happen simultaneously, while losses below a certain threshold are in fact independent.\ Hence, such a dependence structure can have a strong impact on tail risk measures, see Section \ref{se: tail risk measures}, but, from a statistical point of view, it can easily be confused with complete independence unless extreme losses were already observed in the past.
\end{remark}

\section{Monotonicity of $\rh$ and the relation between $\de$ and $\ga$}\label{ sc: Monotonicity}

For two $n$-dimensional random vectors $T^1$ and $T^2$, we write $T^2 \leq_\lo T^1$ --
\textit{lower orthant order} -- if  $\P(T^2\leq a)\leq \P(T^1\leq a)$ for all $a \in \R^n$.\ For a detailed discussion of the lower orthant order and related multivariate stochastic orders, we refer to \cite[Chapter~6]{RushMathRiskAn}. For applications of such comparison criteria in the context of risk aggregation problems, we refer to \cite{Bignozzietal2015} and the references therein.\
For two random vectors with \textit{identical marginals}, i.e., 
$$T^1_i\stackrel{\rm d}{=} T^2_i\quad \text{for }i=1,\ldots,n,$$  the inequality $T^2\leq_\lo T^1$ can be interpreted as a stronger positive dependence among the components of $T^1$ than among the components of $T^2$. In the case $n=2$ and assuming identical marginals for $T^1$ and $T^2$, the lower orthant order coincides with the order of \textit{positive quadrant dependence} $\leq_{\rm PQD}$, \textit{PQD order} for short,
originally introduced in \cite{Lehmann1966} and sometimes called \textit{correlation order}, cf.\ \cite{Dhaene1996}, which is one of the most popular stochastic orders adopted to formalize the intuition of stronger positive dependence between two random variables.\ For instance, all measures of concordance in the sense of \cite{Scarsini84}, such as Spearman's rho and Kendall's tau, are consistent with PQD order, cf.\ \cite{Nelsencop}.\ The same is true for the Pearson correlation, cf.\ \cite[Lemma 2]{Lehmann1966}.\ For an overview and a more detailed discussion of these dependence concepts, we refer to \cite{Nelsencop}.

\begin{proposition}\label{mono1}
Let $0<\gamma_1\leq \gamma_2<1$. Then, $$X^{\gamma_2}\leq_\lo X^{\gamma_1}.$$ 
\end{proposition}

\begin{proof}
For $j=1,2$, let $U^{\gamma_j}$ be defined as the vector whose $i$-th element is  given by
\[
U^{\gamma_j}_i =\gamma_jV_i \mathds{1}_{\{U\leq\gamma_j\}} + U \mathds{1}_{\{U>\gamma_j\}},
\]
where $(V_1,V_2,...,V_n,U)$ are i.i.d.\ uniformly distributed on $(0,1).$

First, we show that $ \gamma_1\leq \gamma_2$ implies $U^{\gamma_2}\leq_\lo U^{\gamma_1}$.\ Since the marginal distributions of $U^{\gamma_1}$ and $U^{\gamma_2}$ are uniform on the interval $(0,1)$, we only have to show that 
$$
\P\big(U^{\gamma_2}\leq a\big)\leq \P\big(U^{\gamma_1}\leq a\big)\quad \text{for all }a\in (0,1)^n.
$$
If $a\in \prod^{n}_{i=1}(\gamma_1, 1)$, then
\begin{align*}
\P\big(U^{\gamma_1}\leq a\big)&=1-\P\bigg(\bigcup_{i=1}^n \big\{U_i^{\gamma_1}> a_i\big\}\bigg)=1-\P\bigg(\bigcup_{i=1}^n \big\{U> a_i\big\}\bigg)=\P\big((U,\ldots, U)\leq a\big)\\
&=\min_{i=1,\ldots, n} a_i\geq \P\big(U^{\gamma_2}\leq a\big).
\end{align*}
Now, let $a\in (0,1)^n \sm \prod^{n}_{i=1}(\gamma_1, 1)$. Then, there exists some $k\in \{1,\ldots, n\}$ with $a_k<\ga_1\leq \ga_2$. Now, let $j\in \{1,2\}$ and $\om\in \Om$ with $U^{\ga_j}(\om)\leq a$. Assume, towards a contradiction, there existed some $i\in \{1,\ldots, n\}$ with $U^{\ga_j}_i(\om)>\ga_j$. Then, $\ga_j<U_i^{\ga_j}(\om)=U(\om)=U^{\ga_j}_k(\om)\leq a_k$, which contradicts $a_k<\ga_1\leq \ga_2$. Hence,
$$\P(U^{\gamma_j}\leq a)=\P\bigg(\bigcap_{i=1}^{n} \big\{U_i^{\gamma_j}\leq a_i\wedge \gamma_j\big\} \bigg) \quad\text{for }j=1,2.$$ 
Moreover, by construction of $U^{\gamma_j}$,
$$\P\bigg(\bigcap_{i=1}^{n} \big\{U_i^{\ga_j}\leq a_i\wedge\gamma_j\big\}\cap \{U> \ga_j\} \bigg)=0 \quad\text{for }j=1,2,$$
which implies that
\begin{align*}
       \P\big(U^{\gamma_j}\leq a\big)&=  \P\bigg(\bigcap_{i=1}^{n} \big\{U_i^{\gamma_j}\leq a_i\wedge\gamma_j\big\}\cap \{U\leq \gamma_j\} \bigg)=\gamma_j  \prod_{i=1}^{n} \P\big(\gamma_j V_i\leq a_i\wedge \gamma_j \big)\\
       &=\gamma_j\prod_{i=1}^{n} \frac{a_i\wedge\gamma_j}{\gamma_j}=a_k\prod_{\substack{{i=1}\\{i\neq k}}}^n \frac{a_i\wedge\gamma_j}{\gamma_j}\quad \text{for }j=1,2.
   \end{align*}
    Since $\frac{a_i\wedge \ga_1}{\ga_1}\geq \frac{a_i\wedge \ga_2}{\ga_2}$ for all $i=1,\ldots, n$, it follows that
    \[
    \P\big(U^{\gamma_1}\leq a\big)= a_k\prod_{\substack{{i=1}\\{i\neq k}}}^n \frac{a_i\wedge\gamma_1}{\gamma_1}\geq a_k\prod_{\substack{{i=1}\\{i\neq k}}}^n \frac{a_i\wedge\gamma_2}{\gamma_2}= \P\big(U^{\gamma_2}\leq a\big).
    \]
  We have therefore shown that $U^{\ga_2}\leq_\lo U^{\ga_1}$.\ Finally, the statement $X^{\gamma_2}\leq_\lo X^{\gamma_1}$ follows from the observation that $X^{\gamma_j}_i=F^{-1}_{\mu_i}\big(U^{\gamma_j}_i \big)$ for $i=1,\ldots,n$ and $j=1,2$ together with the fact that the lower orthant order is preserved under nondecreasing transforms, cf.\ \cite[Theorem~6.G.3]{shaked2007stochastic}.
\end{proof}

\begin{corollary}\label{cor.lo.consistent}
    Let $\cM\subset \cP(\R)$ be nonempty, $\rh\colon \cpl^n(\cM)\to [-1,1]$ be a regular dependence measure consistent with lower orthant order, and  $0<\gamma_1\leq \gamma_2<1$. Then, $$\rh\big(X_1^{\gamma_2},\ldots, X_n^{\gamma_2}\big)\leq \rh\big(X_1^{\gamma_1},\ldots, X_n^{\gamma_1}\big).$$ 
\end{corollary}

As mentioned before, most dependence measures, such as the Pearson correlation, Spearman's rho, and Kendall's tau, are consistent with the PQD order.\ We have the following corollary.

\begin{corollary}\label{cor.pqd.consistent}
    Let $\cM\subset \cP(\R)$ be nonempty, $n\geq 2$, $\rh_{ij}\colon \cpl^2(\cM)\to [-1,1]$ be regular bivariate dependence measures consistent with PQD order for $i,j=1,\ldots, n$ with $i\neq j$, $D\colon [-1,1]^{n(n-1)}\to [-1,1]$ as in Example \ref{ex.abs.dep.meas} c), $\rh\colon \cpl^n(\cM)\to [-1,1]$ be given by \eqref{eq.aggregation}, and  $0<\gamma_1\leq \gamma_2<1$. Then, $$\rh\big(X_1^{\gamma_2},\ldots, X_n^{\gamma_2}\big)\leq \rh\big(X_1^{\gamma_1},\ldots, X_n^{\gamma_1}\big).$$ 
\end{corollary}

\begin{proof}
      By construction, $X^{\gamma_1}$ and $X^{\gamma_2}$ have identical marginals, see \eqref{eq: rap} or property (i) in Theorem \ref{th: main result 1}. From Proposition \ref{mono1}, we know that $X^{\gamma_2}\leq_\lo X^{\gamma_1}$.\ Hence, using the continuity from below of $\P$, for $a\in \R^n$ and $i,j=1,\ldots, n$ with $i\neq j$, we have
        \begin{align*}
    \P\Big(\big\{X_i^{\gamma_2}\leq a_i\big\}\cap \big\{X_j^{\gamma_2}\leq a_j\big\}\Big) &= \lim_{k\to \infty}\P\big(X^{\gamma_2}\leq a^k\big)\leq\lim_{k\to \infty}\P\big(X^{\gamma_1}\leq a^k\big)\\
    &=\P\Big(\big\{X_i^{\gamma_1}\leq a_i\big\}\cap \big\{ X_j^{\gamma_1}\leq a_j\big\}\Big) 
      \end{align*}
      with $(a^k)_{k\in \N}\subset \R^n$ satisfying $a_i^k=a_i$, $a_j^k=a_j$, and $a^k_l\to \infty$ for $l =1,\ldots, n$ with $l\neq i,j$,
      so that   $$(X_i^{\gamma_2},X_j^{\gamma_2})\leq_{\rm PQD}(X_i^{\gamma_1},X_j^{\gamma_1})\quad \text{for all }i,j=1,\ldots,n\text{ with }i\neq j.$$ 
    The consistency of $\rh_{ij}$ with PQD order thus yields that $$\rh_{ij}(X_i^{\gamma_2},X_j^{\gamma_2})\leq \rh_{ij}(X_i^{\gamma_1},X_j^{\gamma_1})\quad \text{for all }i,j=1,\ldots,n\text{ with }i\neq j,$$
    and the statement then follows from the fact that $D$ is nondecreasing.
\end{proof}

\begin{remark}\label{rem: copula}
    If the marginal distributions of $X^\gamma$ are continuous, by Sklar's Theorem, the dependence structure among the elements of $X^\gamma$ can be represented using the unique copula of the random vector.\
    Using the construction in \eqref{eq: rap}, the copula of a vector $X^\gamma$ can retrieved from the joint distributions of the random variables
     \begin{equation}\label{eq: rapunif}
       U^\gamma_i=\gamma V_i \mathds{1}_{\{U\leq\gamma\}} + U \mathds{1}_{\{U>\gamma\}}
   \end{equation}
     with $(V_1,V_2,...,V_n,U)$  vector of independent random variables with uniform distribution on $[0,1]$. Let $c^\gamma$ denote the joint distribution of $(U^\gamma_1,\ldots, U^\gamma_n)$.\ Along the lines of the proof of Proposition~\ref{mono1}, we have shown that, for $a\in (0,1)^n$,
     \begin{equation}\label{eq:copgamma}
        c^\gamma(a)= P(U^{\gamma}\leq a)=\begin{cases}
            \min_{i=1,\ldots,n}{a_i}, \text{ for } a\in \prod^{n}_{i=1}(\gamma, 1), \\
\ga \prod_{i=1}^{n}\frac{a_i\wedge \gamma}{\gamma}, \text{ for }  a\in (0,1)^n \sm \prod^{n}_{i=1}(\gamma, 1).
        \end{cases}
     \end{equation}
Observe that Theorem~\ref{th: main result 1} can then be reformulated in terms of copulas by saying that, for every $\delta\in (0,1)$ and all $\mu_1,\ldots, \mu_n\in \cM$ with continuous distribution functions, there exists some $\ga\in(0,1)$ and an $n$-dimensional random vector $X^\gamma$ with a copula as in \eqref{eq:copgamma} satisfying $\rh(X_1^\ga,\ldots, X_n^\ga)\leq \delta$ and $ X_i^\gamma \sim \mu_i$ for $i=1,\ldots,n$.

\end{remark}

Theorem \ref{th: main result 1} shows that partial knowledge on multiple dependence measures is compatible with an upper comonotonic dependence structure indexed by a parameter $\gamma\in(0,1)$, in which $1-\gamma$ describes the probability of observing a common tail event.\  Nonetheless, one can still wonder, which specific values for $\gamma\in (0,1)$ are compatible with the fact that a regular dependence measure is less or equal than $\delta\in (0,1)$.\ Proposition \ref{mono1} and Corollary \ref{cor.lo.consistent} give a first insight by stating that  there exist a decreasing relationship between $\gamma\in (0,1)$ and $\rh(X^\gamma)$ for regular dependence measures that are consistent with lower orthant order. The following example  provides an answer to this question for the Pearson correlation, cf.\ Example \ref{ex.dep.meas} a), Spearman's rho, cf.\ Example \ref{ex.dep.meas} b), and Kendall's tau, cf.\ Example \ref{ex.dep.meas} c).

\begin{example} \label{ex. correlations}\
\begin{enumerate}
    \item[a)]  \textbf{Pearson correlation.} We start by computing the Pearson correlation in the case of an identical distribution, i.e., $\mu_i=\mu\in \cP(\R)$ with $\int_\R x^2\,\mu(\d x)<\infty$.\ Let $i,j\in\{1,\ldots, n\}$ with $i\neq j$ and $\ga\in (0,1)$.\ Following the first three lines of the computation in \eqref{eq.coupling.xgamma}, for $\ga\in (0,1)$, we find that
    \[
    \E \big(X_i^{\ga}X^\ga_j\big)=\frac1\ga \bigg(\int_0^\ga F_\mu^{-1}(u)\,\d u\bigg)^2+\int_\ga^1 \big(F_\mu^{-1}(u)\big)^2\,\d u.
    \]
    Hence, the condition $\cor(X_i^\ga,X_j^\ga)= \de\in (0,1)$ is equivalent to the equality
\begin{equation}\label{eq.pearson.de.ga}
     \frac1\ga \bigg(\int_0^\ga F_\mu^{-1}(u)\,\d u\bigg)^2+\int_\ga^1 \big(F_\mu^{-1}(u)\big)^2\,\d u -\bigg(\int_\R x\,\mu(\d x)\bigg)^2= \de \var(\mu),
\end{equation}
    where, by Corollary \ref{cor.pqd.consistent}, the left-hand side is decreasing in $\ga$. Hence, by Jensen's inequality and the intermediate value theorem, for each $\de\in (0,1)$, the set $I_\de$ of all $\ga\in (0,1)$, such that \eqref{eq.pearson.de.ga} holds, is a nonempty closed interval. In particular, $\min I_\de$ is the smallest $\ga\in (0,1)$ such that \eqref{eq.pearson.de.ga} is satisfied for $\de\in (0,1)$.\ On the other hand, unless $\mu$ is a Dirac, for each $\ga\in (0,1)$, there exists at most one $\de\in (0,1)$ such that \eqref{eq.pearson.de.ga} holds. Observe that the equation \eqref{eq.pearson.de.ga} for $\ga$ and $\de$ is independent of $n$.
    
    In the sequel, we present two examples, where \eqref{eq.pearson.de.ga} can be solved explicitly. Moreover, the case, where $\mu$ is a uniform distribution on (0,1), is discussed in part b), below.
    \begin{enumerate}
    \item[(i)] Let $\de,\ga\in (0,1)$ and $\mu:=B(1,p)$ be a Bernoulli distribution with $p\in (0,1)$. Then, $F_\mu^{-1}=\eins_{[1-p,1)}$. If $\ga\leq 1-p$, then \eqref{eq.pearson.de.ga} does not admit a solution. Hence, we only consider the case $\ga>1-p$. Then,
    \eqref{eq.pearson.de.ga} simplifies to
    \[
     \frac1\ga (1-\ga-p)^2+(1-\ga)-p^2=\de p(1-p),
    \]
    which gives
    $$\de=\frac{(1-p)(1-\ga)}{p\ga}\quad \text{or, equivalently,}\quad \ga=\frac{1}{1+\de \frac{p}{1-p}}.$$
 Performing similar computations in the general case, where $\mu_i=B(1,p_i)$ with $p_i\in (0,1)$ for $i=1,\ldots, n$, one finds that
 \[
\cor\big(X_i^\ga,X_j^\ga\big)\leq\de\in (0,1) \quad\text{for }i,j=1,\ldots, n\text{ with }i\neq j
 \]
if
\[
 \frac{1}{1+\de \cdot \frac{p_{\min}}{1-p_{\min}}}\leq \ga<1
\]
with $p_{\min}:=\min\{p_1,\ldots, p_n\}$.
\begin{figure}[htb]
\includegraphics[width=10cm ]{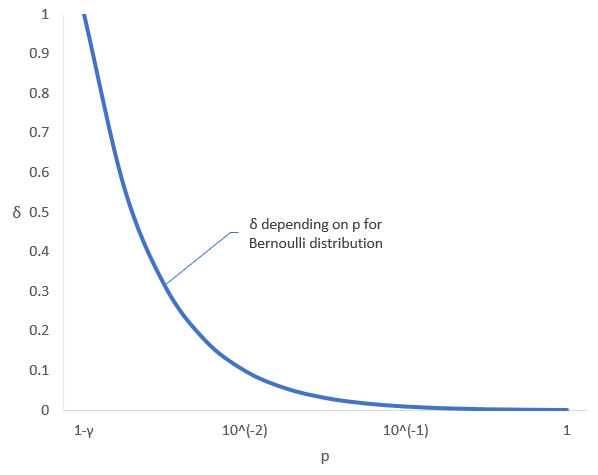}
\caption{$\de$ in dependence of $p$ for $\mu \sim B(1,p)$ and  fixed $\ga=0.999$}
\end{figure}
   \item[(ii)] Now, assume that $\mu={\rm Exp}(\la)$ with $\la>0$, so that $F_\mu^{-1}(u)=-\frac{\ln(1-u)}\la$ for $u\in (0,1)$. In this case, for $\de,\ga\in (0,1)$, \eqref{eq.pearson.de.ga} becomes
   \[
   \frac1{\ga \la^2} \bigg(\int_0^\ga \ln(1-u)\,\d u\bigg)^2+\frac1{\la^2} \int_\ga^1 \big(\ln(1-u)\big)^2\,\d u-\frac1{\la^2}=\frac\de{\la^2}.
   \]
    Dividing by $\la^2$, we find that
    \[
    \frac1\ga \bigg(\int_0^\ga \ln(1-u)\,\d u\bigg)^2+\int_\ga^1 \big(\ln(1-u)\big)^2\,\d u-1=\de,
    \]
    which leads to 
    \[
    \de=(1-\ga)\Bigg(1+\frac{\big(\ln(1-\ga)\big)^2}{\ga}\Bigg).
    \]
    Note that $\de\in (0,1)$ is the same for all $\la>0$.
\begin{figure}[htb]
\includegraphics[width=10cm]{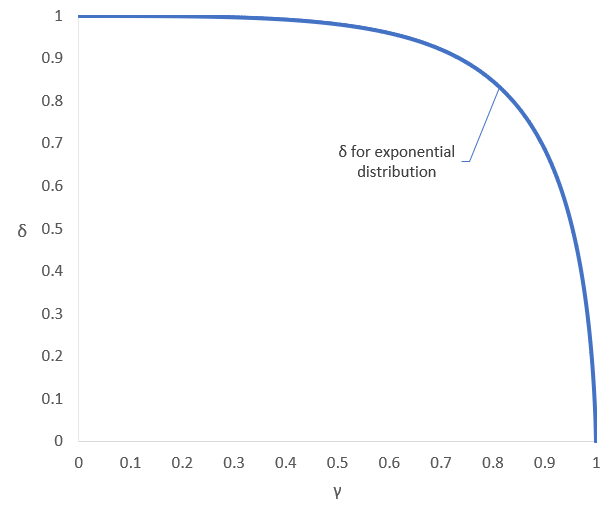}
\caption{$\de$ in dependence of $\ga$ for $\mu \sim {\rm Exp}(\la)$ with arbitrary $\la>0$}
\end{figure}
    
    \end{enumerate}
    
    \item[b)]  \textbf{Spearman's rho.} Let $i,j\in \{1,\ldots, n\}$ with $i\neq j$, and $X^\ga$ as in Remark \ref{rem.rap} for $\ga\in (0,1)$. We now compute the Spearman's rho between $X^\ga_i$ $X^\ga_j$, which is given by $\spear\big(X^\gamma_i,X^\gamma_j\big)= \cor\big(U^\ga_i,U^\ga_j\big)$ with $U^\ga_i$ and $U^\ga_j$ given by \eqref{eq: rapunif}.\ Using part a), we find that
    \[
     \E\big(U^\ga_i U^\ga_j\big)=\frac1\ga \bigg(\int_0^\ga u\,\d u\bigg)^2+\int_\ga^1 u^2\,\d u =\frac\ga4+\frac{1-\ga}3.
    \]
     Hence,
     \[
     \spear (X^\gamma_i,X^\gamma_j)=12\bigg(\frac\ga4+\frac{1-\ga}3-\frac14\bigg)=1-\ga^3.
    \]
Going back to Theorem~\ref{th: main result 1}, we find that $\spear\big(X^\gamma_i,X^\gamma_j\big)= \de$ if and only if $\gamma= (1-\de)^\frac{1}{3}$ for $\de\in (0,1)$, independently of $n$.
\item[c)] \textbf{Kendall's tau.} Consider the same setup as in part b).\ We proceed in a similar manner, using the representation
\begin{align*}
    \tau\big(X^\ga_i,X^\ga_j\big)=4\E\Big(c^\ga \big(U^\ga_i,U^\ga_j\big)\Big)-1
\end{align*}
with $c^\ga$ explicitly given in \eqref{eq:copgamma}.\ Then,
\begin{align*}
    \E\Big(c^\ga \big(U^\ga_i,U^\ga_j\big)\Big)&= \E\Big(c^\ga \big(U^\ga_i,U^\ga_j\big)\eins_{\{U\leq \ga \}}\Big)+\E\Big(c^\ga \big(U^\ga_i,U^\ga_j\big)\eins_{\{U> \ga \}}\Big)\\
    &= \E\Big(c^\ga\big(\ga V_i ,\ga V_j\big) \eins_{\{U\leq \ga \}}\Big)+\E\big(c^\ga(U,U)\eins_{\{U> \ga\}}\big)\\
     &= \ga\E\big(V_iV_j\eins_{\{U\leq \ga \}}\big)+\E\big(\min\{U,U\}\eins_{\{U> \ga\}}\big)\\
   &= \ga^2\E(V_i)\E(V_j)+\E\left(U\eins_{\{U> \ga\}}\right)=\frac{\ga^2}{4}+\frac{1-\ga^2}{2}.
\end{align*}
Thus, 
\[
   \tau(X^\ga_i,X^\ga_i)=1-\ga^2.
\]
In particular, $\tau(X^\ga_i,X^\ga_i)=\de$ if and only if $\ga =(1-\de)^{\frac12}$ for $\de\in (0,1)$, independently of $n$.
\end{enumerate}

\end{example}

\section{Risk aggregation and tail risk measures} \label{se: tail risk measures}

In this section, we put the first main result, Theorem \ref{th: main result 1}, in the context of so-called tail risk measures as introduced in \cite{liu2021theory}. Throughout, we work in the setup of Theorem \ref{th: main result 1}.\ That is, we consider a probability space $(\Om,\cF,\P)$, a uniformly distributed random variable $U\colon \Om\to (0,1)$, a nonempty set $\cM\subset \cP(\R)$, and a regular dependence measure $\rh\colon \cpl^n(\cM)\to [-1,1]$.\ Then, for $n\in \N$, $\de>0$ and $\mu_1,\ldots, \mu_n\in \cM$, let $\ga\in (0,1)$ and $X_1,\ldots, X_n$ be random variables on $\Om$ such that (i) - (iv) in Theorem \ref{th: main result 1} is satisfied.\ As before, we use the notation $X=(X_1,\ldots, X_n)$ and $Y=(Y_1,\ldots, Y_n)$ with $Y_i:=F_{\mu_i}^{-1}(U)$ for $i=1,\ldots, n$.\ Note that $Y$ is a comonotonic vector, i.e., it consists of perfectly correlated random variables, that has the same marginals as $X$.\ We point out that comonotonicity is often seen as the strongest possible form of positive dependence and describes a complete lack of diversification benefits, cf.\ \cite{DHAENE20023} for a detailed discussion of this concept.

Moreover, let $C^\ga:=\{U\leq \ga\}$ and $\P^\ga\colon \cF\cap C^\ga\to [0,1]$ be given by
\[
\P^\ga(A):=\frac{\P(A)}{\ga}\quad\text{for all }A\in \cF\cap C^\ga.
\]
We define $Z_i\colon C^\ga\to \R$ by $Z_i(\om):=X_i(\om)$ for all $\om\in C^\ga$ and $i=1,\ldots, n$. Again, we use the notation $Z=(Z_1,\ldots,Z_n)$.

Let $f\colon \R^n \to \R$ be a nondecreasing and left-continuous function, i.e.,
\[
f(x)\leq f(y)\quad \text{for all }x,y\in \R^n\text{ with }x \leq y
\]
and $f(x)=\lim_{y\uparrow x}f(y)$ for all $x\in \R^n$. The aim of this section is to analyze the value at risk of $f(X)$. In particular, we show that $f(X)$ and $f(Y)$ exhibit the same tail risk starting from the risk threshold $\ga$, i.e., for all $\alpha \in [\ga,1)$
\[
\VaR_\P^\alpha(f(X))= \VaR_\P^\alpha(f(Y)).
\]
In this context, we define 
\[
q(\al):= f \big(F^{-1}_{\mu_1}(\alpha),\ldots,F^{-1}_{\mu_n}(\alpha) \big)\in \R.
\]
Observe that the map $(0,1)\to \R,\;\al\mapsto q(\al)$ is nondecreasing and left-continuous, since $f$ and $F_{\mu_i}^{-1}$ for $i=1,\ldots, n$ are nondecreasing and left-continuous.\
We start with the following three auxiliary results. 
\begin{lemma} \label{lemma 1}
Let $a\in \R$ with $a\geq q({\ga})$. Then, $\P \big(\{f(X)>a \} \cap \{U\leq \ga \} \big)=0$.
\end{lemma}

\begin{proof}
Let $\om\in \Om$ with $f\big(X(\om)\big)>a\geq q({\ga}).$ By definition of $q(\ga)$ and since $f$ is nondecreasing, there exists some $i\in \{1,\ldots,n\}$ such that $X_i(\om) > F^{-1}_{\mu_i}(\ga)$.\ Hence, using property (iv) in Theorem \ref{th: main result 1}, we can conclude that
\begin{align*}
\P \big(\{f(X)>a \} \cap \{U\leq \ga \} \big) &\leq \P \Bigg( \bigcup_{i=1}^{n}\{X_i > F^{-1}_{\mu_i}(\ga) \} \cap \{U\leq \ga \} \Bigg) \\
&\leq \sum_{i=1}^n \P \big( \{X_i > F^{-1}_{\mu_i}(\ga) \} \cap \{U\leq \ga \} \big)=0.
\end{align*}
\end{proof}

\begin{lemma} \label{lemma 2} 
For all $\al\in (0,1)$,
\[
\VaR_\P^\al\big(f(Y)\big)=q(\al)
\]
\end{lemma}

\begin{proof}
Let $\al\in (0,1)$ and observe that $f\big(Y(\om)\big)> q(\al)$ implies that $U(\om)>\al$ for all $\om\in \Om$ by monotonicity of $f$. Hence,
\[
\P\big(f(Y)> q(\al)\big)\leq \P(U>\al)\leq 1-\al,
\]
so that $\VaR_\P^\al\big(f(Y)\big)\leq q(\al)$.\ Now, let $a\in \R$ with $a<q(\al)$.\ Then, by left-continuity of $q$, there exists some $\be\in (0,\al)$ with $a<q(\be)\leq q(\al)$. Hence,
\[
1-\al<1-\be=\P(U\geq \be)\leq \P\big(f(Y)\geq q(\be)\big)\leq \P\big(f(Y)>a\big),
\]
which shows that $\VaR_\P^\al\big(f(Y)\big)= q(\al)$.
\end{proof}

\begin{lemma}\label{lem.var.small}
For all $\al\in (0,1)$,
\[
\VaR_{\P^\ga}^{\al}\big(f(Z)\big)=\VaR_\P^{\al\ga}\big(f(X)\big).
\]
Moreover,
\[
\lim_{\al\uparrow 1}\VaR_{\P^\ga}^{\al}\big(f(Z)\big)=q(\ga).
\]
\end{lemma}

\begin{proof}
Let $\al\in (0,1)$. Then, by defintion of $Z$ and $\P^\ga$, for all $a\in \R$,
\[
\P\big(f(X)\leq a\big)\geq \P\big(\{f(X)\leq a\}\cap \{U\leq \ga\}\big) =\ga \P^\ga\big(f(Z)\leq a\big),
\]
which shows that $$\VaR_\P^{\al\ga}\big(f(X)\big)\leq \VaR_{\P^\ga}^{\al}\big(f(Z)\big).$$
Since, by property (iv) in Theorem \ref{th: main result 1}, 
$$\P^\ga\big(Z_i\leq F_{\mu_i}^{-1}(\ga)\big)=1\quad\text{for }i=1,\ldots,n$$
 and $f$ is nondecreasing, it follows that $$\VaR_{\P^\ga}^{\al}\big(f(Z)\big)\leq q(\ga).$$
Now, let $a\in \R$ with $a<\VaR_{\P^\ga}^{\al}\big(f(Z)\big)$.\ 
Then, $a<q(\ga)$, so that, by monotonicity of $f$, $f\big(X(\om)\big)\leq a$ implies that $U(\om)\leq \ga$ for all $\om\in \Om$. Hence,
\[
\P\big(f(X)\leq a\big)=\P\big(\{f(X)\leq a\}\cap\{U\leq \ga\}\big)=\ga \P^\ga\big(f(Z)\leq a\big)<\al \ga,
\]
which implies that $a\leq \VaR_{\P}^{\al\ga}\big(f(X)\big)$ and, consequently, $\VaR_{\P}^{\al\ga}\big(f(X)\big)\geq \VaR_{\P^\ga}^{\al}\big(f(Z)\big)$. It remains to show that $q(\ga)\leq \lim_{\al\uparrow 1}\VaR_{\P^\ga}^{\al}\big(f(Z)\big).$ To that end, observe that
\[
\lim_{\al\uparrow 1}\VaR_{\P^\ga}^{\al}\big(f(X)\big)=\inf\big\{a\in \R\,\big|\, \P^\ga\big(f(Z)>a\big)=0\big\}.
\]
Now, let $a\in \R$ with $a<q(\ga)$.\ Then, using the left-continuity of $q$, there exists some $\be\in (0,\ga)$ with $a<q(\be)\leq q(\ga)$.\ Hence, using property (iii), i.e., the independence of the random variables $Z_1,\ldots, Z_n$ on $C^\ga$ under $\P^\ga$, and property (iv) in Theorem \ref{th: main result 1},
\begin{align*}
\P^\ga\big(f(Z)>a\big)&\geq \P^\ga\big(f(Z)\geq q(\be)\big)\geq \prod_{i=1}^n\P^\ga\big(Z_i\geq F_{\mu_i}^{-1}(\be)\big)\geq 
\bigg(1-\frac{\be}{\ga}\bigg)^{n}>0,
\end{align*}
so that
$$\VaR_\P^\ga\big(f(X)\big)=\inf\big\{a\in \R\,\big|\, \P^\ga\big(f(Z)>a\big)=0\big\}\geq q(\ga).$$
 \end{proof}

\begin{theorem}\label{th: var.equality}
For all $a\in \R$ with $a\geq q({\ga})$,
$$
\P\big(f(Y)\leq a\big)=\P\big(f(X)\leq a\big).$$
Moreover, for all $\al\in (0,\ga)$,
\[
\VaR_\P^{\al}\big(f(X)\big)=\VaR_{\P^\ga}^{\frac\al\ga}\big(f(Z)\big)\geq \VaR_{\P^\ga}^{\al}\big(f(Z)\big)
\]
and, for all $\al\in [\ga,1)$,
\[
\VaR_\P^\al\big(f(X)\big)=\VaR_\P^\al\big(f(Y)\big)=q(\al).
\]
\end{theorem}

\begin{proof}
We first prove that $\VaR_\P^\ga\big(f(X)\big)\geq q(\ga)$.\ By Lemma \ref{lem.var.small}, it follows that
\[
\VaR_\P^\ga\big(f(X)\big)=\lim_{\al\uparrow 1}\VaR_\P^{\al\ga}\big(f(X)\big)=\lim_{\al\uparrow 1}\VaR_{\P^\ga}^{\al}\big(f(X)\big)=q(\ga).
\]
Now, let $a\in \R$ with $a\geq q({\ga})$ and $\om\in \Om$ with $f\big(Y(\om)\big)>a\geq q({\ga})$.\ Then, by monotonicity of $f$, we can conclude that $U(\om)>\ga$ and thus $f\big(X(\om)\big)=f\big(Y(\om)\big)>a$. Using Lemma \ref{lemma 1}, it follows that
\begin{align*}
 \P \big(\{f(X)>a \}\big) &=  \P \big(\{f(X)>a \} \cap \{U> \ga \} \big)\\
 &=\P \big(\{f(Y)>a \} \cap \{U> \ga \} \big)= \P \big(\{f(Y)>a \}\big).
\end{align*}
Hence, for all $\al\in [\ga,1)$,
\begin{align*}
 \VaR^\al\big(f(X)\big)&=\inf\big\{ a\in \big[q(\ga),\infty\big) \,\big|\, \P\big(f(X)> a\big)\leq 1-\al\big\}\\
 &=\inf\big\{ a\in \big[q(\ga),\infty\big)  \,\big|\, \P\big(f(Y)> a\big)\leq 1-\al\big\}=\VaR^\al\big(f(Y)\big).
\end{align*}
The remaining statements now follow from Lemma \ref{lemma 2} and Lemma \ref{lem.var.small}.
\end{proof}

We briefly recall the notion of a tail risk measure, introduced in \cite{liu2021theory}, in a slightly modified version that is specific to our setup.\ Let $\cX$ be a nonempty set of random variables.\ For any random variable $T\in \cX$ with distribution $\nu$, let $U_Z$ be a uniformly distributed random variable with $T=F_\nu^{-1}(U_T)$ $\P$-almost surely.\ The existence of such a random variable follows, for example, from \cite[Lemma A.32]{MR3859905}, recalling that $(\Om,\cF,\P)$ is implicitly assumed to be atomless via the existence of a uniformly distributed random variable. Then, the tail risk of $Z$ beyond its $\al$-quantile is defined as
\[
T^\al:=F_\nu^{-1}(\al+(1-\al)U_T)\quad\text{for }\al\in (0,1).
\]
We point out that, a priori, the definition of $Z^\al$ also depends on the choice of the uniformly distributed random variable $U_Z$. However, passing from $U_Z$ to another uniformly distributed random variable, say $V$, it follows that
\begin{equation}\label{eq.uniqueness.tail.risk}
 T^\al\stackrel{\rm d}{=}F_\nu^{-1}(\al+(1-\al)V) \quad\text{for all }\al\in (0,1).
\end{equation}
 Therefore, $Z^\al$ is unique up to equality in distribution.

Following \cite[Definition 1]{liu2021theory}, for $\al\in (0,1)$, we say that a map $R\colon \cX\to \R$ is an \textit{$\al$-tail risk measure} if $R(T_1)=R(T_2)$ for all $T_1,T_2\in \cX$ with $T_1^\al\stackrel{d}{=}T_2^\al$.

\begin{theorem}\label{th: tailequivalent}
Let $\cX$ be a set of random variables containing $f(X)$ and $f(Y)$ and $\al\in [\ga,1)$. Then, for every  $\al$-tail risk measure $R\colon \cX\to \R$,
\[
R\big(f(X)\big)=R\big(f(Y)\big).
\]
\end{theorem}

\begin{proof}
  Let $T_1:=f(X)$ and $T_2:=f(Y)$ with distributions $\nu_1\in \cP(\R)$ and $\nu_2\in \cP(\R)$, respectively. Then, by Theorem \ref{th: var.equality}, it follows that
 \begin{equation}\label{eq.tail.distribution}
 F_{\nu_1}^{-1}(\al)=\VaR_\P^\al\big(f(X)\big)=\VaR_\P^\al\big(f(Y)\big)=F_{\nu_2}^{-1}(\al) \quad\text{for all }\al\in [\ga,1).
 \end{equation}
 Then, by \eqref{eq.uniqueness.tail.risk}, $T_i^\al\stackrel{\rm d}{=} F_{\nu_i}^{-1}(\al+(1-\al)U)$ for $i=1,2$. Hence, by \eqref{eq.tail.distribution},
 \[
  T_1^\al \stackrel{\rm d}{=} F_{\nu_1}^{-1}(\al+(1-\al)U)=F_{\nu_2}^{-1}(\al+(1-\al)U)\stackrel{\rm d}{=}T_2^\al\quad\text{for all }\al\in [\ga,1).
 \]
 The statement now follows directly from the definition of an $\al$-tail risk measure.
\end{proof}
Theorem \ref{th: tailequivalent} establishes an equivalence between comonotonicity and upper comonotonicity for arbitrary tail risk measures and nondecreasing left-continuous aggregation functions.\
In this regard, observe that Theorem \ref{th: tailequivalent} generalizes \cite[Proposition~6]{Cheug2009uc}, where only the special case $ f(x)=\sum_i^nx_i$ is considered, and the risk measure is either $\text{VaR}^\alpha$, $\text{TVaR}^\alpha$, or $\text{ES}^\alpha$ with $\alpha\in (\gamma,1)$, cf.\ \cite{Cheug2009uc} for the details.\ We underline that, even in the case $f(x)=\sum_i^nx_i$, considered in \cite[Proposition~6]{Cheug2009uc}, Theorem~\ref{th: tailequivalent} offers an additional contribution with respect to \cite{Cheug2009uc}, in that our result also covers the limiting case $\alpha=\gamma$. 

Consider the case, where $\cX$ is the set of \textit{all} random variables on $\Om$. Then, Theorem \ref{th: tailequivalent} states that, for $\al\in [\ga,1)$ and every $\al$-tail risk measure $R$ on $\cX$, the worst possible $\al$-tail risk of $f(X)$, taken over all random vectors $X$ with $X_i\sim \mu_i$ for $i=1,\ldots, n$ and $\rh(X_1,\ldots, X_n)\leq \de$, exceeds the value $R\big(f(Y)\big)$, i.e., the $\al$-tail risk associated to a perfectly comonotone realization of the marginal distributions.

This becomes even more apparent if one considers coherent $\al$-tail risk measures and weighted sums of the components of the random vector $X$.\ Since every tail risk measure is law-invariant, we write $R(\nu)=R(T)$ for any random variable $T\in \cX$ with $T\sim \nu\in \cP(\R)$. 

\begin{corollary}\label{cor:tail risk.pos.hom}
Let $\mu_i=\mu\in \cP(\R)$ for $i=1,\ldots, n$, $\cX$ be a convex cone of random variables containing $F_\mu^{-1}(U)$ and $X_i$ for $i=1,\ldots, n$, $\al\in [\ga,1)$, $R\colon \cX\to \R$ be a positively homogeneous $\al$-tail risk measure, and $\la_1,\ldots, \la_n\in [0,\infty)$. Then,
\[
R\bigg(\sum_{i=1}^n\la_i X_i\bigg)=\sum_{i=1}^n \la_iR(X_i)= R(\mu)\sum_{i=1}^n\la_i.
\]
\end{corollary}

\begin{proof}
 By Theorem \ref{th: tailequivalent}, it follows that
 \[
 R\bigg(\sum_{i=1}^n\la_i X_i\bigg)=R\bigg(\sum_{i=1}^n\la_i F_\mu^{-1}(U)\bigg)=R(\mu)\sum_{i=1}^n\la_i=\sum_{i=1}^n\la_iR(X_i).
 \]
\end{proof}

Consider again the case, where $\cX$ is the set of all random variables on $\Om$. Then, the previous corollary tells us that, for $\al\in [\ga,1)$ and every coherent $\al$-tail risk measure $R$ on $\cX$, the worst possible $\al$-tail risk of $\sum_{i=1}^n \la_i X_i$ with $\la_i\geq0$ for $i=1,\ldots, n$, taken over all random vectors $X$ with $X_i\sim \mu$ for $i=1,\ldots, n$ and $\rh(X_1,\ldots, X_n)\leq \de$, is exactly $R(\mu)\sum_{i=1}^n\la_i$.\footnote{Indeed, since $R$ is coherent, $R\big(\sum_{i=1}^n\la_iX_i\big)\leq \sum_{i=1}^n\la_i R(X_i)=R(\mu)\sum_{i=1}^n\la_i$ for any random vector $X$ with $X_i\sim \mu$ for $i=1,\ldots, n$.} 

The third main result of this section is another application of Theorem \ref{th: main result 1} in the context of expectiles, cf.\ \cite{MR4216328}.\ For a random variable $T$ on $\Om$ with $\E(|T|)<\infty$ and $\al\in (0,1)$, the \textit{$\al$-expectile} ${\rm ex}^\al(T)\in \R$ of $Z$ is defined as the unique solution to the equation
\[
\al\E\Big(\big(T-{\rm ex}^\al(T)\big)_+\Big)=(1-\al)\E\Big(\big(T-{\rm ex}^\al(T)\big)_-\Big).
\]
Note that, for $\al=\frac12$, ${\rm ex}^\al(T)=\E(T)$.\ It is well-known that ${\rm ex}^\al$ is a law-invariant risk measure, so that we also use the notation ${\rm ex}^\al(\nu):={\rm ex}^\al(T)$ if $T\sim \nu$.\ We refer to \cite{bellini2017risk,MR3145849, delbaen2013expectiles,MR3479327} for a detailed and axiomatic study of expectiles as well as their interpretation in a financial context.

As a consequence of \cite[Theorem 3]{MR4216328}, we have the following theorem.

\begin{theorem}\label{thm: expectiles}
 Assume that $\mu_i=\mu\in \cP(\R)$ for $i=1,\ldots, n$ with $\int_\R |x|\,\mu(\d x)<\infty$, and let $\la_1,\ldots \la_n\in [0,\infty)$.\ Then, for all $\al\in \big[\frac12,1\big)$ with ${\rm ex}^\al(\mu)\geq F_\mu^{-1}(\ga)$,
 \[ {\rm ex}^\al\bigg(\sum_{i=1}^n\la_iX_i\bigg)=\sum_{i=1}^n\la_i{\rm ex}^\al(X_i)={\rm ex}^\al(\mu)\sum_{i=1}^n\la_i.
 \]
\end{theorem}

\begin{proof}
 For $\al=\frac12$ the statement is clear since, in this case, the $\al$-expectile is the mean.\ Let $\al\in \big(\frac12,1\big)$ with ${\rm ex}^\al(\mu)\geq F_\mu^{-1}(\ga)$.\ By \cite[Theorem 3]{MR4216328}, we have to show that
\begin{equation}\label{eq.thm.expectiles}
  \P\Big(\big(\la_i X_i-{\rm ex}^\al(\la_iX_i)\big)\big(\la_j X_j-{\rm ex}^\al(\la_jX_j)\big)<0\Big)=0\quad\text{for }i,j=1,\ldots, n. 
 \end{equation}
 Since ${\rm ex}^\al$ is law-invariant and positively homogeneous, and $X_i\sim \mu$ for $i=1,\ldots, n$, this is equivalent to showing that
\[
\P\Big(\big(X_i-{\rm ex}^\al(\mu)\big)\big(X_j-{\rm ex}^\al(\mu)\big)<0\Big)=0 \quad\text{for }i,j=1,\ldots, n. 
\]
To that end, first observe that
\begin{align*}
\P\Big(\big(X_i-{\rm ex}^\al(\mu)\big)\big(X_j-{\rm ex}^\al(\mu)\big)<0\Big)&\leq \P\big(\{X_i<{\rm ex}^\al(\mu)\}\cap \{X_j>{\rm ex}^\al(\mu)\}\big)\\
&\quad +\P\big(\{X_i>{\rm ex}^\al(\mu)\}\cap \{X_j<{\rm ex}^\al(\mu)\}\big) 
\end{align*}
for $i,j=1,\ldots, n$.\ Hence, it remains to show that
\[
\P\big(\{X_i<{\rm ex}^\al(\mu)\}\cap \{X_j>{\rm ex}^\al(\mu)\}\big)=0\quad\text{for }i,j=1,\ldots, n.
\]
Since ${\rm ex}^\al(\mu)\geq F_\mu^{-1}(\ga)$, by property (iv) in Theorem \ref{th: main result 1}, it follows that
\begin{align*}
\P\big(\{X_i<{\rm ex}^\al(\mu)\}\cap & \{X_j>{\rm ex}^\al(\mu)\}\big)=\P\big(\{X_i<{\rm ex}^\al(\mu)\}\cap \{X_j>{\rm ex}^\al(\mu)\}\cap \{U>\ga\}\big)\\
&\; =\P\big(\big\{F_\mu^{-1}(U)<{\rm ex}^\al(\mu)\big\}\cap \big\{F_\mu^{-1}(U)>{\rm ex}^\al(\mu)\big\}\cap \{U>\ga\}\big)=\P(\emptyset)=0
\end{align*}
for $i,j=1,\ldots, n$. We have therefore proved the validity of \eqref{eq.thm.expectiles}, and the statement now follows from \cite[Theorem 3]{MR4216328}.
\end{proof}

\begin{remark} \label{remark quotient}
Consider the case, where $\mu_i=\mu\in \cP(\R)$ for $i\in \N$ with $\int_{\R}x^2\,\mu(\d x)<\infty$ and $m^\ga:=\frac{1}{\ga}\int_0^\ga F_\mu^{-1}(u)\,\d u>0$. Moreover, let $\al\in [\ga,1)$ and $(\la_i)_{i\in \N}\subset [0,\infty)$ with 
$$\liminf_{n\to \infty}\frac1n\sum_{i=1}^n\la_i>0\quad\text{and}\quad\lim_{n\to\infty}\frac{1}{n^2}\sum_{i=1}^n \la_i^2=0.\footnote{These conditions are, for example, satisfied if $\liminf_{i\to \infty}\la_i>0$ and $\lim_{i\to \infty}\frac{\la_i}{i}= 0$.}$$
In this remark, we aim to study the asymptotic behaviour of the ratio
\[
r^\al(n):=\frac{\VaR^\al\big(\sum_{i=1}^n\la_i X_i\big)}{\VaR_{\P^\ga}^\al\big(\sum_{i=1}^n\la_i Z_i\big)}
\]
as the number of customers $n$ tends to infinity.

First observe that $Z_1,\ldots, Z_n$ are i.i.d.\ with distribution $\mu^\ga \in \cP(\R)$, mean $m^\ga$, and variance $\var(\mu^\ga)<\infty$.\ Let $\ep>0$, then, using Markov's inequality and the independence of $Z_1,\ldots, Z_n$,
\[
\P^\ga\bigg( \sum_{i=1}^n\la_i \big(Z_i-m^\ga\big)>n\ep\bigg)\leq \frac{\var(\mu^\ga)}{\ep^2} \frac{1}{n^2}\sum_{i=1}^n\la_i^2.
\]
Since $\lim_{n\to\infty}\frac{1}{n^2}\sum_{i=1}^n \la_i^2=0$, there exist $c>0$ and $n_0\in \N$ such that $\sum_{i=1}^n\la_i\geq \frac1c$ and
\[
\P^\ga\bigg( \sum_{i=1}^n\la_i Z_i>n\ep+m^\ga\sum_{i=1}^n\la_i\bigg)=\P^\ga\bigg( \sum_{i=1}^n\la_i \big(Z_i-m^\ga\big)>n\ep\bigg)\leq 1-\al
\]
for all $n\in \N$ with $n\geq n_0$.\ Hence, for all $n\in \N$ with $n\geq n_0$,
\[
\VaR_{\P^\ga}^\al\bigg(\sum_{i=1}^n\la_i Z_i\bigg)\leq n\ep+m^\ga\sum_{i=1}^n\la_i\leq \big(m^\ga+c\ep\big) \sum_{i=1}^n\la_i,
\]
so that, by Theorem \ref{cor:tail risk.pos.hom},
\[
 r^\al(n)\geq \frac{F_{\mu}^{-1}(\al)}{m^\ga+c\ep}\quad\text{for all }n\in \N \text{ with }n\geq n_0 .
\]
Since $F_{\mu}^{-1}$ is nondecreasing, we have therefore shown that
\begin{equation}\label{eq.varestimate1}
\limsup_{n\to \infty} \frac{\VaR^\al\big(\sum_{i=1}^n\la_i X_i\big)}{\VaR_{\P^\ga}^\al\big(\sum_{i=1}^n\la_i Z_i\big)}\geq \frac{\VaR^\al(\mu)}{\frac{1}{\al}\int_0^\al F_\mu^{-1}(u)\,\d u}\geq 1.
\end{equation}
In particular, there exists some $n_0\in \N$ such that
\begin{equation}\label{eq.varestimate2}
 \frac{\VaR^\al\big(\sum_{i=1}^n\la_i X_i\big)}{\VaR_{\P^\ga}^\al\big(\sum_{i=1}^n\la_i Z_i\big)}\geq \frac{\VaR^\al(\mu)}{\int_{\R} x\,\mu(\d x)}
\quad\text{for all }n\in \N\text{ with }n\geq n_0.\footnote{Indeed, if $\mu$ is a Dirac, $r(n)=1$ for all $n\in \N$ and $\VaR^\ga(\mu)=\int_\R x\,\mu(\d x)$.\ Otherwise, $F_\mu^{-1}$ nonconstant and nondecreasing, so that $m^\ga=\frac1\ga\int_0^\ga F_{\mu}^{-1}(u)\,\d u<\int_0^1 F_{\mu}^{-1}(u)\,\d u=\int_{\R} x\,\mu(\d x)$.}
\end{equation}
\end{remark}

\begin{example}[Value at risk for credit portfolio] \label{ex.application} Tail risk measures such as value at risk or expected shortfall are an essential part of risk management for most financial institutions. In this example, we discuss the implications of our theoretical results on dependence uncertainty 
for the value at risk in the context of portfolio credit loss for a bank caused by defaulting borrowers.\ 

Throughout, we use the same notation as in the proof of Theorem  \ref{th: main result 1} and consider a simplified setting with $n$ borrowers, where the credit loss caused by a default equals the remaining exposure of the respective borrower, i.e., for $i=1,\ldots, n$, the loss of borrower $i$ can be interpreted as a random variable
$$L_i= \text{Exposure}_i\cdot X_i\quad\text{with}\quad X_i\sim B(1,p_i),$$
where $p_i\in (0,1)$ is the probability of default (PD) of borrower $i$ over a one year time horizon.\ Note that, in practice, the true PD might of course differ from the estimated PD of the rating model.\ However, in order to isolate the
impact of 
dependence uncertainty on the portfolio value at risk, we neglect the possibility of uncertain PDs, and implicitly assume $p_i$ to be the true PD of borrower $i$ for $i=1,\ldots, n$. For an axiomatic study of model uncertainty related to PDs, including applications in the context of regulatory capital requirements, we refer to \cite{nendel2023axiomatic}. 

For $\pi\in \cpl^n(\cM)$ with $\cM$ consisting of all $\mu\in \cP(\R)$ such that $\int_\R x^2\,\mu(\d x)<\infty$, we consider
\[
\rh(\pi):=\max_{i\neq j} \cor\big(\pi^{ij}\big)
\]
where $\pi^{ij}$ is given as in Example \ref{ex.abs.dep.meas} c).
Our aim is now to calculate the value of risk at level $\alpha=0.999$ of 
\[
f\big(X_1,\ldots,X_n\big):= \sum_{i=1}^n \text{Exposure}_i \cdot X_i,
\]
allowing for possible dependences up to a certain level $\de\in (0,1)$, i.e., $\rh(X_1,\ldots, X_n)\leq \de$.
Observe that, if $\la_i=\text{Exposure}_i\geq 0$ for $i=1,\ldots, n$, the function $f\colon \R^n \to \R$, given by $f(x)=\sum_{i=1}^n \la_i x_i$ for $x\in \R^n$, is nondecreasing and continuous.\ We point out that the rather high value for $\alpha$ is the confidence level commonly used in practice for the value at risk over a one year horizon, describing a once in a thousand years event. From Example \ref{ex. correlations} a) i), we know that if
\begin{equation}\label{eq.estimate.ex.var}
\ga:=\al=0.999 \geq\frac{1}{1+\de\frac{p_{\min}}{1-p_{\min}}},
\end{equation}
 the conditions of Theorem \ref{th: tailequivalent} are met. In this case, $\ga>1-p_{\min}\geq 1-p_i$ and, on the event $C^{\ga}$, the random variables $Z_1,\ldots, Z_n$ are independent under $\P^{\ga}$ with $\P^{\ga}\big( Z_i=0 \big)=\frac{1-p_i}{\ga}$  and $\P^{\ga}\big( Z_i=1 \big)=\frac{\ga-1+p_i}{\ga}$ for $i=1,\ldots, n$. Hence,
 \[
 Z_i \sim B\left(1,\frac{\ga-1+p_i}{\ga} \right)\quad\text{for }i=1,\ldots,n.
 \]
For the sake of simplicity, we now consider a portfolio consisting of $n=1000$ borrowers, each having an exposure of $1$ and a PD of $1\%$, but the results can easily be transferred to more complex portfolio constellations with different exposures and different PDs as previously described.\ 
Choosing $\de=\frac1{10}$, one sees that \eqref{eq.estimate.ex.var} is satisfied.\ In this case,
\[
f(Z)=\sum_{i=1}^n Z_i \sim B\left(n,\frac{\ga-1+p}{\ga} \right),
\]
so that the conditional value at risk on the event $C^{\ga}$ is given by 
\[
\VaR^\al_{\P^{\ga}} \big(f(Z) \big)=20.
\]
On the other hand, using Theorem \ref{th: main result 1} and Theorem \ref{th: var.equality}, we obtain
\[
\VaR^\al \big(f(X) \big)=1000=50\cdot \VaR^\al_{\P^{\ga}}\big(f(Z)\big),
\]
i.e., the potential value at risk for a once in a thousand years event could be $50$ times higher than the value at risk, given that we have not observed such an event before. By \eqref{eq.varestimate1}, we know that
\[
 \limsup_{n\to \infty} \frac{\VaR^\al \big(f(X) \big)}{\VaR_{\P^\ga}^\al \big(f(Z) \big)}\geq \frac{0.999}{0.999-0.99}=111.
\]
The threshold
\[
 \frac{\VaR^\al\big(B(1,p)\big)}{p}=100,
\]
given in \eqref{eq.varestimate2}, is reached for $n_0=100\,000$ customers.\ The asymptotic behaviour of the ratio $r^{0.999}$ is depicted in Figure \ref{var.behaviour}.

In applications, one usually only has a default rate history of a few decades.\ Consequently, with a very high probability one only observes default rates on the event $C^{\ga}$, where in our example the defaults behave independently.\ Information about tail behavior for $\al\geq 0.999$ is usually not available and needs to be estimated.\ Simply extrapolating correlation effects from the available data history, in this case $C^{\ga},$ might not be sufficient to estimate the value at risk for such large values of $\al$, but could lead to a drastic underestimation of the true risk, in our case, possibly by a factor of $50$.

Last but not least, we point out that, by Theorem \ref{th: var.equality}, we also have
\[
\VaR^{\al}\big(f(X)\big)=\VaR^{\frac{\al}{\ga}}_{\P^{\ga}}\big(f(Z)\big)\geq \VaR^{\ga}_{\P^{\ga}}\big(f(Z)\big)\text{ for }\al<0.999.
\]
In Figure \ref{VaR2} both values coincide on the chosen grid for the displayed values of $\al<0.999$.

\begin{figure}[htb]\label{var.behaviour}
\includegraphics[width=11cm]{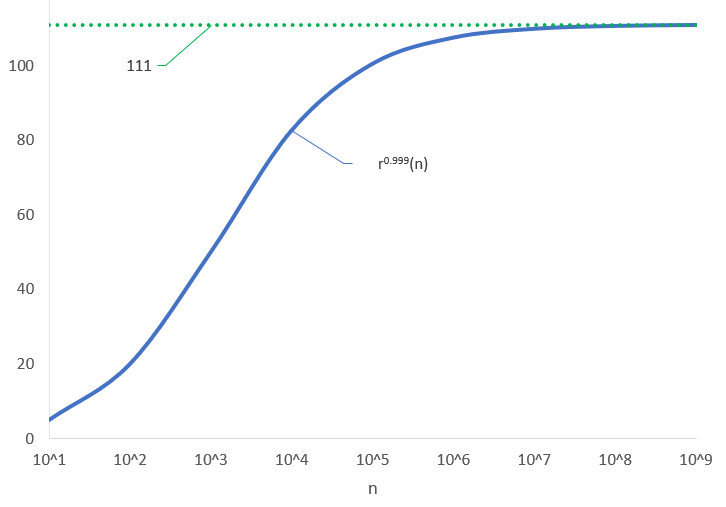}
\caption{Quotient between $\VaR^{0.999}\big(f(X)\big)$ and $\VaR^{0.999}_{\P^{\ga}}\big(f(Z)\big)$  in dependence of the number of customers $n$}
\label{Quotient VaR}
\end{figure}

\begin{figure}[htb]
\includegraphics[width=11cm]{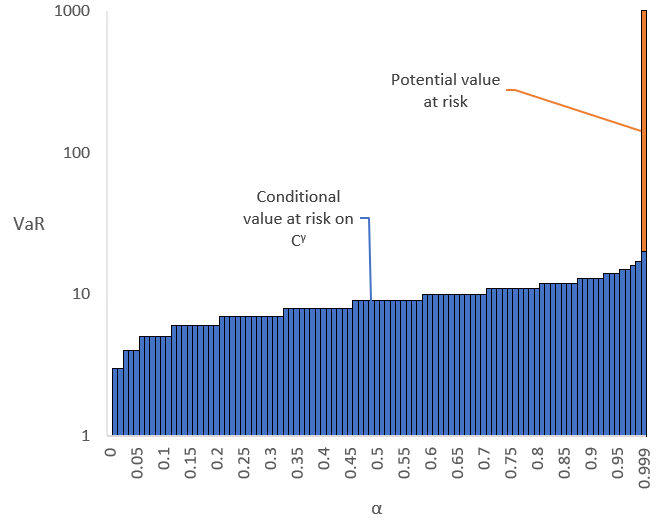}
\caption{Comparison between potential value at risk and conditional value at risk on $C^{\ga}$}
\label{VaR2}
\end{figure}

\end{example}

\appendix

\section{Some auxiliary results}\label{app.A}

For the proof of the first main result, we need the following standard result from optimal transport.

\begin{lemma}\label{lem.compactness.couplings}
Let $\mu_1,\ldots, \mu_n\in \cP(\R)$. Then, the set $\cpl(\mu_1,\ldots, \mu_n)$ is weakly compact.
\end{lemma}

\begin{proof}
 Let $\ep>0$. Then, there exists a compact set $K\subset \R$ with
 \[
  \mu_i(\R\sm K)<\frac{\ep}n \quad\text{for all }i=1,\ldots, n.
 \]
 Hence, for any $\pi\in \cpl(\mu_1,\ldots, \mu_n)$,
 \[
  \pi (\R^n\sm K^n)\leq \sum_{i=1}^n \mu_i(\R \sm K)<\ep.
 \]
 We have therefore shown that the set $\cpl(\mu_1,\ldots, \mu_n)$ is tight. Moreover,
\[
\cpl(\mu_1,\ldots, \mu_n)=\bigcap_{i=1}^n \big\{\pi\in \cP(\R^n)\,\big|\, \pi\circ \pr_i^{-1}=\mu_i\big\}
\]
is weakly closed since $\pr_i\colon \R^n\to \R$ is continuous for all $i=1,\ldots, n$.
\end{proof}

The following lemma is sort of a folklore theorem for copulas.\ For the reader's convenience, we state and prove it in our setup.

\begin{lemma}\label{lem.unif.convergence.couplings}
Let $\mu_1,\ldots, \mu_n\in \cP(\R)$ with continuous distribution functions $F_{\mu_1},\ldots, F_{\mu_n}$ and $(\pi^k)_{k\in \N}\subset \cpl(\mu_1,\ldots, \mu_n)$ be a sequence of couplings with $\pi^k\to \pi\in \cP(\R^n)$ in distribution as $k\to \infty$. Then, there exist unique copulas $c$ and $(c^k)_{k\in \N}$ such that
\begin{equation}\label{eq.copula.representation}
  \pi\big((-\infty, a]\big)=c\big(F_{\mu_1}(a_1),\ldots, F_{\mu_n}(a_n)\big)\quad\text{and}\quad \pi^k\big((-\infty, a]\big)=c^k\big(F_{\mu_1}(a_1),\ldots, F_{\mu_n}(a_n)\big)
\end{equation}
for all $k\in \N$ and $a=(a_1,\ldots, a_n)\in \R^n$. Moreover, 
\[
\sup_{u\in [0,1]^n}\big|c^k(u)- c(u)\big|\to 0\quad \text{as }k\to \infty.
\]
In particular,
\[
 \sup_{a\in \R^n}\big|\pi^k\big((-\infty, a]\big)-\pi\big((-\infty, a]\big)\big|\to 0 \quad \text{as }k\to \infty.
\]
\end{lemma}

\begin{proof}
By Lemma \ref{lem.compactness.couplings}, it follows that $\pi\in \cpl(\mu_1,\ldots, \mu_n)$. Since $F_{\mu_1},\ldots, F_{\mu_n}$ are continuous, by Sklar's theorem, there exist unique copulas $c$ and $(c^k)_{k\in \N}$ such that \eqref{eq.copula.representation} is satisfied. Now, let $a\in \R^n$. Again, since $F_{\mu_1},\ldots, F_{\mu_n}$ are continuous, 
\[
\pi\bigg(\bigcup_{j=1}^n \pr_j^{-1}\big(\{a_j\}\big)\cap (-\infty,a]\bigg)\leq \sum_{j=1}^n \pi\circ\pr_j^{-1}\big(\{a_j\}\big)=\sum_{i=1}^n \mu_j\big(\{a_j\}\big)=0.
\]
Hence, by the portmanteau theorem \cite[Corollary 8.2.10]{MR2267655}, it follows that
\[
c^k\big(F_{\mu_1}(a_1),\ldots, F_{\mu_n}(a_n)\big)=\pi^k\big((-\infty, a]\big)\to \pi\big((-\infty, a]\big)=c\big(F_{\mu_1}(a_1),\ldots, F_{\mu_n}(a_n)\big) 
\]
as $k\to \infty$. Since $F_{\mu_1},\ldots, F_{\mu_n}$ are continuous, this implies that
\begin{equation}\label{eq.pointwise.conv.couplings}
\lim_{k\to \infty}c^k(u)=c(u)\quad\text{for all }u\in [0,1]^n.
\end{equation}
  Now, let $\ep>0$ and $m\in \N$ with $\frac{n}{m}\leq\frac\ep2$. Let $u_i:=\frac{i}m$ for $i=0,\ldots,m$. Then, 
  \[
  c(u_{i_1+1},\ldots,u_{i_n+1})-c(u_{i_1},\ldots,u_{i_n})\leq \sum_{j=1}^n \big(u_{i_j+1}-u_{i_j}\big)\leq \frac{n}{m}<\frac\ep2
  \]
  for all $i_1,\ldots, i_n\in \{0,\ldots, m-1\}$.\ Moreover, by \eqref{eq.pointwise.conv.couplings}, there exists some $k_0\in \N$ such that
 \[
 \big|c^k(u_{i_1},\ldots, u_{i_n})-c(u_{i_1},\ldots, u_{i_n})\big|<\frac\ep2
 \]
 for all $k\in \N$ with $k\geq k_0$ and $i_1,\ldots, i_n\in \{0,\ldots, m\}$.
 
 Now, let $u\in [0,1]^n$, $k\in \N$ with $k\geq k_0$, and $i_1,\ldots,i_n\in \{0,\ldots, m-1\}$ with 
 $$u_{i_j}\leq u_j\leq u_{i_j+1}\quad\text{for all }j=1,\ldots, n.$$
 Then,
 \[
 c^k(u_{i_1},\ldots, u_{i_n})-c(u)\leq c^k(u)-c(u)\leq c^k(u_{i_1+1},\ldots, u_{i_n+1})-c(u)\quad\text{for }k\in \N.
 \]
 Since 
 \begin{align*}
  c(u)-c^k(u_{i_1},\ldots, u_{i_n})&< c(u)-c(u_{i_1},\ldots, u_{i_n})+\frac\ep2\\
  &\leq c(u_{i_1+1},\ldots, u_{i_n+1})-c(u_{i_1},\ldots, u_{i_n})+\frac\ep2< \ep
 \end{align*}
 and
 \begin{align*}
   c^k(u_{i_1+1},\ldots, u_{i_n+1})-c(u)&<\frac\ep2+c(u_{i_1+1},\ldots, u_{i_n+1})-c(u)\\
   &\leq \frac\ep2+ c(u_{i_1+1},\ldots, u_{i_n+1})-c(u_{i_1},\ldots, u_{i_n})< \ep,
 \end{align*}
 it follows that $\big|c^k(u)-c(u)\big|<\ep$ for all $k\in \N$ with $k\geq k_0$.
\end{proof}


\end{document}